\documentclass[11pt]{article}

\usepackage[margin=1in]{geometry}
\usepackage[T1]{fontenc}
\usepackage{lmodern}
\usepackage{amsmath,amssymb,amsthm,mathtools}
\usepackage{microtype}
\usepackage{setspace}
\usepackage[authoryear,round]{natbib}
\usepackage{endnotes}
\usepackage{xcolor}
\usepackage[colorlinks=true,allcolors=blue!55!black]{hyperref}

\newtheorem{lemma}{Lemma}
\newtheorem{proposition}{Proposition}
\newtheorem{corollary}{Corollary}
\theoremstyle{definition}
\newtheorem{assumption}{Assumption}
\newtheorem{definition}{Definition}

\newcommand{\HH}{\mathrm{H}}
\newcommand{\Atech}{\mathrm{A}}
\newcommand{\profit}{\pi}
\newcommand{\E}{\mathbb{E}}

\title{The Reverse Big Push:\\Generative AI and Self-Fulfilling Automation}
\author{%
  Soumen Banerjee\\
  \small Central University of Finance and Economics
  \and
  Jianguo Wang\\
  \small Renmin University of China}
\date{November 2026}
\hypersetup{
  pdftitle={The Reverse Big Push: Generative AI and Self-Fulfilling Automation},
  pdfauthor={Soumen Banerjee and Jianguo Wang},
  pdfkeywords={generative AI, augmentation, automation, strategic complementarity, local demand}
}

\begin{document}

\maketitle

\begin{abstract}
Service firms using generative AI choose between keeping higher-margin human teams or replacing them with pay-per-use AI. Teams commit fixed payroll, whereas AI fees flow to external providers. Since local wages and resident profits sustain demand, each firm's payroll helps fund others' sales. We show that this demand link creates a self-fulfilling trap: a big push in reverse. Even with flexible wages and local ownership, two stable outcomes can coexist: all firms keep teams, or all automate. Every firm earns more when all retain teams, yet none gains by rebuilding a team alone. When firms reorganize over time, expectations decide the outcome: from the same workforce, firms automate if they expect shrinking demand and rebuild if they expect growth. Whenever both outcomes coexist, keeping teams maximizes welfare, though rebuilding lost teams is costly. We characterize the optimal transition: depending on the state, the economy should rebuild at every opportunity, automate at every opportunity, or preserve its current workforce. Recovery can be optimal even when adding a team initially lowers surplus, because each rebuilt team enlarges the market for the next. Finally, a self-financing transfer from automated to team-retaining firms implements this optimal path without public funds and disappears once teams are restored.
\end{abstract}

\noindent\textbf{Keywords:} generative AI, augmentation, automation, strategic complementarity, local demand.

\section{Introduction}

For many service firms, adopting generative AI is no longer the central strategic question. Law practices, architectural and residential design studios, accounting offices, and advisory firms already have access to frontier models. The harder decision is how to organize work around them: keep a team of specialists working alongside the model, or let the model take over a substantial share of the team's work. Retaining specialists improves customization, verification, and accountability, allowing the firm to command a higher price and earn a wider margin on each dollar of revenue. Salaries, however, have to be paid whether client demand is strong or weak. Automated capability has the opposite cost structure: an external developer bears the fixed cost of training the model and charges downstream users by the query or token, so a firm buys automated capability as orders arrive. Taken in isolation, the firm's problem is a familiar trade-off between margin and operating leverage: a busy market pays for a full team, while a thin market rewards cutting payroll and buying capability by usage.

Across a shared local market, however, firms do not face that trade-off in isolation. When a law office or design studio retains its specialists, their wages are spent largely in the local economy, including on services sold by other local businesses. When firms replace part of their teams with usage-priced AI, local spending leaks away: inference fees flow to an external model supplier, and when firm owners live outside the region, payroll savings paid out as dividends leave as well. Unless displaced workers immediately replace their lost earnings locally, each firm's staffing choice alters the size of every other firm's market. Two self-sustaining outcomes can emerge. If firms across the region keep their teams, their combined payrolls generate enough local demand for every team to cover its cost. If firms automate instead, local demand contracts until shedding payroll becomes each firm's best response to weak sales, reproducing the slump that justified automation in the first place.

This is the big-push mechanism of \citet{murphy1989industrialization} run in reverse. In their model, firms that industrialize and pay factory wages enlarge one another's consumer markets, making industrialization profitable when undertaken together. Here firms already use the frontier technology, and collective automation erodes the demand needed to sustain the human teams working alongside it. We call the resulting low-employment trap the \emph{reverse big push}. Whenever both outcomes are stable equilibria, every incumbent firm earns strictly higher profit when all firms retain their teams than when all automate, yet any single firm that rebuilt its team alone in an automated market would lose money. Automation in this range is \emph{defensive}: firms shed specialists not because automation outperforms the old equilibrium, but because lower overhead reduces their break-even scale in a market shrunk by other firms' layoffs.

We develop these results in a regional economy where service incumbents compete against generic AI output, limit-pricing up to their quality advantage as in \citet{grossman1991quality}, while purchasing model inference from an external developer that amortizes its training cost through a markup on usage \citep{romer1990endogenous}. We then ask if falling wages undo the coordination problem. When labor supply slopes upward, widespread automation drives down the local wage. Cheaper labor makes teams less costly to hire, which works against automation, but lower wages also cut workers' spending, which deepens the demand contraction. As long as labor supply is sufficiently elastic, both equilibria survive, greater wage flexibility narrows their coexistence range, and firm profits remain strictly higher under universal augmentation even though wages are higher there too.

We also show that keeping profits in local hands does not eliminate the problem. When residents own shares in local businesses, payroll savings remain inside the region as dividend income rather than flowing to outside shareholders. Even so, automated firms send a larger share of every sales dollar out of the region to pay the upstream model supplier. The coordination trap can arise at any resident ownership share, and under the cutoff conditions in Proposition~\ref{prop:ownership}, higher local ownership widens the range of technologies over which both equilibria coexist.

We next turn to dynamics and ask how an economy moves between deployment regimes over time. Firms reorganize at staggered dates rather than all at once, so a firm choosing its workflow today has to judge what the local market will look like until its next opportunity to adjust. Over an interval of inherited workforces, both recovery and an automation cascade are perfect-foresight equilibria: starting from the same share of surviving teams, firms rebuild if they expect demand to recover and automate if they expect demand to keep shrinking. History fixes the inherited workforce, while expectations determine which steady state the economy approaches \citep{krugman1991history,matsuyama1991increasing}. In a tractable benchmark with zero switching costs and Brownian profitability shocks, applying the selection results of \citet{frankel2000resolving} and \citet{burdzy2001fast} yields a sharp comparative static: when firms revise frequently, improvements in unattended AI capability can tip the selected outcome into automation while universal augmentation remains an equilibrium with higher profit and surplus.

We finally study welfare and optimal policy. Whenever both equilibria coexist, universal augmentation maximizes the joint surplus of workers and firm owners as well as resident welfare, yet individual firms ignore the revenue their payrolls generate for other local businesses. Once teams have been dismantled, however, rebuilding them requires recruitment, training, and the restoration of organizational routines. We solve the social planner's dynamic problem globally under the same staggered revision opportunities and switching costs that firms face. Although social surplus is convex in the share of augmenting firms, the optimal policy takes a simple form across all feasible plans: rebuild at every opportunity above an upper workforce threshold, let automation proceed below a lower threshold when too few teams remain to repay the rebuilding bill, and preserve the inherited workforce in between (with the inaction band vanishing when switching costs are small). Recovery can be optimal even where adding a team initially reduces flow surplus, because each rebuilt team expands the market for teams restored later. Under a condition on surplus curvature, the e-companion shows that the same three-region structure holds under resident ownership, where reorganization costs temporarily reduce local dividends and demand.

Implementing this optimal transition requires no public funds. We construct a state-dependent, budget-balanced transfer within the service sector under which automated firms pay an assessment distributed directly to team-retaining firms. Regulators already move money between firms in this way: under California's Zero-Emission Vehicle regulation and China's dual-credit policy, automakers that sell too few electric vehicles buy credits from automakers that sell more. With outside ownership, setting the transfer equal to the demand externality aligns each firm's lifetime return to retaining a team with the planner's valuation at every date; the e-companion derives the corresponding state-dependent transfer under resident ownership. The payments sum to zero across firms, leave local consumer spending unchanged, and drop to zero once full augmentation is reached, where the high-employment equilibrium sustains itself. Where multiple equilibrium paths remain under the transfer, full implementation pairs the transfer with coordination of expectations on the planner's preferred path.

\subsection{Related literature}

Our paper connects firm-level research on human-AI workflows to market-level models of demand spillovers. A growing literature in information systems and organizational economics asks how AI alters work and authority inside the firm: how organizational coordination shapes the returns to AI prediction \citep{agrawal2024adoption}, how automation interacts with delegated decision rights \citep{dogan2024strategic}, how AI-facilitated knowledge transfer can weaken high performers' incentives and lower firm productivity \citep{yuan2026backfiring}, and when tasks should be automated rather than performed in human-AI teams \citep{fugener2026roles}. Field experiments underscore the commercial value of human involvement, showing that automation can free employees to specialize \citep{gong2024automation}, that keeping a human specialist in the loop increases client reliance and realized consumer value relative to standalone AI \citep{yang2026advisor}, and that AI investments create the greatest firm value when paired with dispersed domain expertise \citep{tambe2026reskilling}. Unlike technology-adoption models in which network effects govern whether firms adopt a new system \citep{katz1986technology}, every firm in our setting already uses the frontier model. We show that \emph{how} firms deploy that model links businesses across unrelated service lines because each firm's payroll helps finance the others' sales.

The closest paper is \citet{falk2026layoff}, in which firm-level automation erodes shared consumer demand because each firm captures its own payroll savings while bearing only a fraction of the aggregate spending loss. They note that while this externality mirrors the demand channel of the big push, automation in their model is a strictly dominant strategy and the equilibrium is unique. We show that the same demand channel generates true strategic complementarity once firms choose between human teams paid as fixed overhead and automated capability bought by usage; their dominant-strategy outcome corresponds to our regime where demand feedback is too weak to overturn the private cost ranking. Both policy and ownership conclusions change as a result. Where automation is dominant, preserving human work requires an ongoing tax or subsidy, whereas in our coordination regime a temporary, budget-balanced industry transfer suffices because universal augmentation is self-sustaining once restored. While recycling capital income to local owners narrows over-automation in \citet{falk2026layoff}, consistent with broader analyses of how consumption depends on who holds automation rents \citep{bayraktar2026automation}, resident ownership here can widen the coordination interval because usage fees leave the region regardless of who owns local equity. Related macroeconomic models generate multiple AI paths through learning-by-using, in which AI improves by learning from the workers who use it \citep{wang2026artificial}, or through agency costs in frictional labor search \citep{fahn2026bad}, or obtain inefficient automation from slow worker reallocation and borrowing constraints \citep{beraja2025inefficient} and uncertain upstream AI research \citep{unveren2023ai}. Focusing directly on downstream spending at a given model capability lets us rank coexisting equilibria by firm profit and solve globally for optimal dynamic recovery, inaction, and automation.

More broadly, our analysis builds on classic models of aggregate-demand complementarities and macroeconomic coordination \citep{diamond1982aggregate,blanchard1987monopolistic,cooper1988coordinating,matsuyama1995complementarities}. Whereas \citet{alvarez2026technology} examine coordinated adoption of a frontier technology with fixed adoption costs, firms here buy access to an existing model by usage and face a coordination problem over how much of their human teams to keep. Relative to task-based frameworks of automation and reinstatement \citep{acemoglu2018race,acemoglu2019automation,acemoglu2020wrong} and policies that steer upstream innovation \citep{korinek2026steering}, we show that downstream product demand governs which deployment mode is privately viable for an already available model, and that an industry-funded transfer can guide deployment without public revenue.

The rest of the paper proceeds as follows. Section~\ref{sec:model} sets up the regional economy, and Section~\ref{sec:static} characterizes the static deployment equilibria, the firm-profit ranking, and the roles of wage flexibility and ownership. Section~\ref{sec:dynamics} analyzes forward-looking reorganization dynamics. Section~\ref{sec:welfare} ranks stationary allocations, solves the dynamic planner's problem, and establishes budget-balanced implementation. Section~\ref{sec:extensions} examines fixed model fees, worker re-employment, offensive alongside defensive automation, and firm heterogeneity, and the e-companion collects the remaining proofs together with equilibrium selection, temporary support for recovery, the resident-ownership planner problem, and per-job labor.

\section{Firms, households, and local demand}
\label{sec:model}

The economy is a region, such as a city or a local labor market, that shares its currency and interest rate with the rest of the country, as in the regional analysis of \citet{nakamura2014fiscal}. It contains a unit mass of local service sectors, indexed by $i\in[0,1]$, and trades an outside good that serves as the numeraire. Four groups appear in the accounts: specialized workers who live in the region; the owners of its service firms, some of whom live elsewhere; the supplier of a frontier generative-AI model, located outside the region; and clients outside the region who buy local services. The model supplier is \emph{upstream}; the service firms that buy access to the model are \emph{downstream}.

\subsection{Households}

There is a mass $\bar N$ of specialized workers, indexed by $z\in[0,\bar N]$. A worker who consumes the local-service composite $Q$ and the outside good $Q_0$, and who supplies $n\in\{0,1\}$ units of labor to the local service sector, obtains utility
\begin{equation}
    U=\frac{Q^{\mu}Q_0^{1-\mu}}{\mu^{\mu}(1-\mu)^{1-\mu}}-c(z)\,n,
    \qquad
    \ln Q=\int_0^1\ln y_i\,di,
    \label{eq:utility}
\end{equation}
where $\mu\in(0,1)$ and $y_i$ is the number of quality units of service $i$ that the worker buys. A job of quality $\lambda$ supplies $\lambda$ quality units, so services sold by different suppliers in the same sector are perfect substitutes after adjusting for quality. The participation cost $c(z)$ is the money value of forgone home production and the disutility of work. The function $c$ is continuously differentiable, nonnegative, and weakly increasing. All residents, including those who hold shares in local firms, have the same preferences over $Q$ and $Q_0$.

Let $p_i/\lambda_i$ be the lowest price per quality unit available in sector $i$, and let $P$ be the geometric mean of these prices across sectors. A household with income $I$ spends $\mu I$ on local services, divided equally across sectors, and $(1-\mu)I$ on the outside good; its indirect utility is $I/P^{\mu}-c(z)n$. A worker therefore participates if and only if $w\geq P^{\mu}c(z)$, whatever the worker's other income. Lemma~\ref{lem:limitpricing} below shows that $P=1$ in every equilibrium. The workers who participate at wage $w$ are then those with $c(z)\leq w$, and market clearing gives the inverse labor-supply schedule
\begin{equation}
    w=w(L)\equiv c(L), \qquad w'(L)\geq 0,
    \label{eq:wagecurve}
\end{equation}
where $L$ is local employment. The wage falls when automation reduces labor demand. The opportunity cost of local employment is
\begin{equation}
    C(L)=\int_0^L w(s)\,ds,
    \label{eq:laborcost}
\end{equation}
and workers' aggregate participation rent is $w(L)L-C(L)$.

\subsection{The model supplier}

The frontier model is produced outside the region by a supplier that pays a one-time training cost $\mathcal F$ and supplies inference at a constant marginal cost $\underline\tau$ per unit of usage. It sells usage to service firms in many regions at a price $\tau>\underline\tau$. Let $\mathcal U$ denote expected discounted lifetime usage across those regions. The supplier trains the model when the expected present value of its usage margin covers the training cost,
\begin{equation}
    (\tau-\underline\tau)\,\mathcal U\geq\mathcal F,
    \label{eq:supplier}
\end{equation}
as the owner of a design does in \citet{romer1990endogenous}. \citet{qian2026economics} study how such a supplier prices fine-tuning and inference when it sells to competing downstream firms. The region is small relative to the supplier's market, so its own usage affects neither $\tau$ nor the decision to train. A downstream firm's payments for the model therefore scale with the jobs it completes, and the training cost is a fixed cost of the supplier.

\subsection{Service firms and prices}

Each sector contains one incumbent firm and a competitive fringe. The fringe sells the generic service that the model produces without a specialized team. Units are chosen so that the generic service has quality one and costs one per job in payments for model usage; competition among fringe suppliers holds its price at one. The incumbent has client relationships and knowledge of its sector that allow it to deliver a better service. It organizes production in one of two modes:
\begin{itemize}
    \item \emph{Augmentation} (mode $\HH$): the incumbent keeps a full team of $\ell<\bar N$ workers who work with the model. It delivers quality $\lambda_{\HH}$ and pays $v_{\HH}>0$ per job for model usage.
    \item \emph{Automation} (mode $\Atech$): the model takes over the tasks of a share $d\in(0,1]$ of the team, and the incumbent retains $(1-d)\ell$ workers. It delivers quality $\lambda_{\Atech}$ and pays $v_{\Atech}>0$ per job for model usage.
\end{itemize}
The team is an overhead: its size does not vary with the number of jobs the firm completes, as with the fixed labor requirement of the modern technique in \citet{murphy1989industrialization}. Their technique also uses labor in proportion to output. In Section~\ref{app:perjob} of the e-companion we add such per-job labor to both modes and show that, with a fixed wage, the equilibrium, welfare, and policy results hold once operating margins are measured net of per-job wages; for given margins and market size, the local spending of those wages widens the range of technologies that supports both equilibria. Employment quantities are densities with respect to the unit mass of sectors, so a set of sectors of measure $dx$ operating in mode $\HH$ employs $\ell\,dx$ workers.

\begin{assumption}[Quality advantage]
\label{ass:quality}
$v_j<\lambda_j$ for $j\in\{\Atech,\HH\}$, and $v_{\HH}/\lambda_{\HH}<v_{\Atech}/\lambda_{\Atech}$.
\end{assumption}

The first part says that in either mode the incumbent's cost per quality unit is below the generic service's price. The second says that the team raises quality by proportionally more than it raises model costs per job. When the two modes use the same model inputs per job, $v_{\HH}=v_{\Atech}$, the second part reduces to $\lambda_{\HH}>\lambda_{\Atech}$: human involvement raises the value of the service to clients \citep{yang2026advisor}. When the team uses fewer model inputs per job, a service of equal quality already satisfies it.

Buyers in each sector purchase from the supplier with the lowest price per quality unit. Every buyer spends a fixed amount on each sector, so the incumbent faces unit-elastic demand up to the price at which buyers switch to the fringe.

\begin{lemma}[Limit pricing]
\label{lem:limitpricing}
Let $M$ denote spending on each sector and $w$ the wage. In either mode the incumbent charges $p_j=\lambda_j$, serves the entire sector, and earns
\begin{align}
    \profit_{\HH}(M,w)&=b_{\HH}M-w\ell, \label{eq:profitH2}\\
    \profit_{\Atech}(M,w)&=b_{\Atech}M-w(1-d)\ell, \label{eq:profitA2}
\end{align}
where $b_j\equiv1-v_j/\lambda_j\in(0,1)$. The price per quality unit equals one in every sector, so $P=1$.
\end{lemma}

\begin{proof}
At a price $p_j>\lambda_j$ the incumbent's price per quality unit exceeds the fringe's, and it sells nothing. At $p_j\leq\lambda_j$ it serves the sector, with indifferent buyers purchasing from the incumbent as is usual under Bertrand competition; it sells $M/p_j$ jobs and earns $M-v_jM/p_j$ before payroll, which rises with $p_j$. Hence $p_j=\lambda_j$, and the price per quality unit is $p_j/\lambda_j=1$ in every sector.
\end{proof}

The incumbent's price exceeds its marginal cost, $p_j=\lambda_j>v_j$, because buyers pay for its quality advantage over the generic service. This is the limit-pricing rule of \citet{murphy1989industrialization}, with the advantage taking the form of quality, as for the quality leader in \citet{grossman1991quality}. The coefficient $b_j$ is the share of each sales dollar that remains in the region as wages and profit after the incumbent pays for model usage. Its difference across modes,
\begin{equation}
    \kappa\equiv b_{\HH}-b_{\Atech}
    =\frac{v_{\Atech}}{\lambda_{\Atech}}-\frac{v_{\HH}}{\lambda_{\HH}}>0,
    \label{ass:singlecrossing}
\end{equation}
is positive by Assumption~\ref{ass:quality}. Augmentation earns more from each additional dollar of local spending because it pays out less of that dollar for model usage; automation saves the payroll of $d\ell$ workers. Condition~\eqref{ass:singlecrossing} is the single-crossing property behind defensive automation: the gain from augmentation rises with the size of the market.

The price of usage does not enter these margins. Let $u_f$ and $u_j$ be the units of usage that a job requires for the generic service and in mode $j$. The normalization above sets $\tau u_f=1$, so $v_j=u_j/u_f$ and
\begin{equation}
    b_j=1-\frac{u_j}{\lambda_j u_f}.
    \label{eq:bjusage}
\end{equation}
The fringe buys usage from the same supplier, so the margin through which the supplier recovers its training cost raises the incumbent's costs and its limit price in the same proportion. It leaves $b_{\HH}$, $b_{\Atech}$, and the single-crossing condition unchanged.

Since the price per quality unit equals one in every sector, local prices adjusted for quality do not depend on how many incumbents automate. The share of income that households spend on the outside good therefore does not move with the allocation, and the same holds under any homothetic preferences over $Q$ and $Q_0$.

\subsection{Local demand and external balance}

Residents hold a share $\theta\in[0,1]$ of the equity of every incumbent; the remaining share is held outside the region. Owners outside the region and the model supplier spend their receipts outside the region, while residents spend a share $\mu$ of all income on local services. Residents have identical homothetic preferences, so the distribution of equity among them does not affect spending.

Let $x\in[0,1]$ denote the \emph{augmentation share}, the measure of incumbents in mode $\HH$, and let $b(x)=b_{\Atech}+\kappa x$. Employment, the wage bill, and aggregate operating profit are
\begin{align}
    L(x)&=\ell(1-d+dx), \label{eq:employment}\\
    W(x)&=w(L(x))L(x), \label{eq:wagebill}\\
    \Pi(x)&=b(x)M-W(x). \label{eq:profitagg}
\end{align}
Let $\overline M>0$ denote spending on local services by clients outside the region. It is the region's export of services and does not depend on the local allocation, because the region is small. Imports are residents' purchases of the outside good, a share $1-\mu$ of their income. Spending on each sector equals outside demand plus residents' local spending out of wages and resident-owned profits:
\begin{equation}
    M=\overline M+\mu\bigl[W(x)+\theta\Pi(x)\bigr].
    \label{eq:demand}
\end{equation}
Substituting equation~\eqref{eq:profitagg} gives the market size
\begin{equation}
    M(x)=\frac{\overline M+\mu(1-\theta)W(x)}{1-\mu\theta b(x)}.
    \label{eq:marketx}
\end{equation}
With outside owners, $\theta=0$ and $M(x)=\overline M+\mu W(x)$. With resident owners who spend like workers, $\theta=1$ and $M(x)=\overline M/[1-\mu b(x)]$.

Every sales dollar is paid out as wages, as payments for model usage, or as profit. Summing over sectors and using equation~\eqref{eq:demand} gives the region's external balance:
\begin{equation}
    \overline M=(1-\mu)\bigl[W(x)+\theta\Pi(x)\bigr]+\bigl[1-b(x)\bigr]M(x)+(1-\theta)\Pi(x).
    \label{eq:accounts}
\end{equation}
The region's service exports pay for residents' purchases of the outside good, for model usage, and for dividends to owners outside the region. Equation~\eqref{eq:accounts} holds in every allocation.\endnote{An external transfer $T$ to residents adds $\mu T$ to the numerator of equation~\eqref{eq:marketx}. In the external balance, receipts become $\overline M+T$ and residents' imports become $(1-\mu)[W(x)+\theta\Pi(x)+T]$. Model payments and dividends are evaluated at the resulting market size.}

When an additional incumbent augments, the wage bill rises by $W'(x)=d\ell[w(L)+Lw'(L)]$ and local value added rises by $\kappa M(x)$. Residents receive the extra wages, give up the resident-owned share of the profit that pays for them, and receive the resident-owned share of the extra value added. Differentiating equation~\eqref{eq:marketx},
\begin{equation}
    M'(x)=\frac{\mu(1-\theta)W'(x)+\mu\theta\kappa M(x)}{1-\mu\theta b(x)}>0.
    \label{eq:Mprime}
\end{equation}
With outside owners, only the payroll channel operates: the payroll that automation saves leaves the region as dividends. With resident owners, payroll is a transfer among residents, and the channel is the additional payment for model usage, which leaves the region. When residents own part of the firms ($\theta>0$), the local multiplier $1/[1-\mu\theta b(x)]$ amplifies the effect. Market size rises with augmentation for every ownership share.

\subsection{Equilibrium}

An incumbent is atomistic and takes $M$ and $w$ as given. By Lemma~\ref{lem:limitpricing}, its gain from augmentation is
\begin{equation}
    G(x)=\profit_{\HH}\bigl(M(x),w(L(x))\bigr)-\profit_{\Atech}\bigl(M(x),w(L(x))\bigr)
    =\kappa M(x)-d\ell w\bigl(L(x)\bigr).
    \label{eq:gain2}
\end{equation}

\begin{assumption}[Viability]
\label{ass:viability}
$b_{\Atech}M(x)\geq(1-d)\ell\, w(L(x))$ for every $x\in[0,1]$.
\end{assumption}

An automated incumbent therefore covers the payroll of its retained team in every allocation and never cedes its sector to the fringe. The condition holds automatically when $d=1$.

\begin{definition}
An equilibrium is a share $x\in[0,1]$ of incumbents in mode $\HH$, a wage $w$, a market size $M$, and prices such that (i) participation and spending decisions are optimal given prices and incomes; (ii) each incumbent's price and mode maximize its profit, taking $M$ and $w$ as given; (iii) the labor market clears, $w=w(L(x))$; and (iv) spending on each sector satisfies equation~\eqref{eq:demand}.
\end{definition}

By Lemma~\ref{lem:limitpricing} and Assumption~\ref{ass:viability}, a share $x$ is an equilibrium if and only if $G(x)\geq0$ whenever a positive measure of incumbents augments and $G(x)\leq0$ whenever a positive measure automates. The remaining analysis uses only this characterization.

Wage adjustment and the demand channel both appear in the slope of the gain:
\begin{equation}
    G'(x)=\kappa M'(x)-d^2\ell^2 w'\bigl(L(x)\bigr).
    \label{eq:Gprime}
\end{equation}
The first term is the increase in local spending created by one more augmenting firm, valued at augmentation's margin advantage $\kappa$. The second is the stabilizing effect of a higher wage on the relative cost of augmentation. The maintained complementarity condition is
\begin{equation}
    \kappa M'(x)>d^2\ell^2 w'\bigl(L(x)\bigr)
    \quad\text{for all }x\in[0,1].
    \label{ass:complementarity}
\end{equation}
With outside owners it reads $\kappa\mu[w(L)+Lw'(L)]>d\ell w'(L)$ for all $L\in[\ell(1-d),\ell]$. The condition makes $G$ strictly increasing without fixing the wage. It holds automatically when local labor supply is perfectly elastic and remains satisfied with an upward-sloping wage curve whenever the demand channel is strong enough.

\section{Deployment equilibria}
\label{sec:static}

In the \emph{automated allocation} every incumbent automates, so $x=0$ and employment is $L_{\Atech}=\ell(1-d)$; in the \emph{augmented allocation} every incumbent augments, so $x=1$ and $L_{\HH}=\ell$. These are the two uniform allocations. Let $w_{\Atech}=w(L_{\Atech})$ and $w_{\HH}=w(L_{\HH})$ be the corresponding wages, and let $M_{\Atech}=M(0)$ and $M_{\HH}=M(1)$ be the corresponding market sizes. The endpoint incentives
\begin{equation}
    G_{\Atech}=G(0)=\kappa M_{\Atech}-d\ell w_{\Atech},
    \qquad
    G_{\HH}=G(1)=\kappa M_{\HH}-d\ell w_{\HH}
    \label{eq:endpointgains}
\end{equation}
measure what a single incumbent gains by keeping its full team when every other incumbent has automated, and when every other incumbent has kept its team.

\begin{proposition}[Flexible-wage equilibrium classification]
\label{prop:flexclassification}
Suppose conditions~\eqref{ass:singlecrossing} and \eqref{ass:complementarity} hold. Apart from equality cases:
\begin{enumerate}
    \item if $G_{\Atech}>0$, augmentation is the unique equilibrium;
    \item if $G_{\HH}<0$, automation is the unique equilibrium;
    \item if $G_{\Atech}<0<G_{\HH}$, both $x=0$ and $x=1$ are equilibria and there is a unique interior equilibrium $x^*\in(0,1)$ satisfying $G(x^*)=0$.
\end{enumerate}
Under myopic best-response adjustment, the endpoints in part 3 are locally stable and $x^*$ is unstable.
\end{proposition}

We call the third case the coordination case; it is the reverse big push. In the automated allocation the local market is too small to repay a single firm for rebuilding its team, whereas in the augmented allocation the same quality margin, earned on a larger market, covers the higher wage bill. A lower wage makes augmentation cheaper in the automated allocation, which nonetheless remains an equilibrium whenever the saving is smaller than the accompanying contraction in local demand.

\subsection{Profits across deployment equilibria}
\label{sec:profits}

The augmentation equilibrium supports a larger market and a weakly higher wage. Let
\begin{equation}
    \Delta M=M_{\HH}-M_{\Atech}>0,
    \qquad
    \Delta w=w_{\HH}-w_{\Atech}\geq0.
\end{equation}
In the coordination case, the increase in sales more than covers the higher wage bill of an automated firm's retained workers:
\begin{equation}
    b_{\Atech}\Delta M>(1-d)\ell\Delta w.
    \label{ass:ranking}
\end{equation}
This inequality follows from maintained viability and the endpoint incentives $G_{\Atech}<0<G_{\HH}$.

\begin{proposition}[Firm-payoff ranking]
\label{prop:firmranking}
Suppose the economy is in the coordination case. Then every incumbent earns strictly more in the augmentation equilibrium than in the automation equilibrium:
\begin{equation}
    \profit_{\HH}(M_{\HH},w_{\HH})>
    \profit_{\Atech}(M_{\Atech},w_{\Atech}).
    \label{eq:flexprofitranking}
\end{equation}
\end{proposition}

An incumbent earns more by automating in the augmented allocation than in the automated allocation, and its equilibrium choice of augmentation raises its profit further. Widespread automation therefore lowers every incumbent's profit, even though each firm's decision to automate is individually optimal in the smaller market.

\subsection{Ownership and wage flexibility}

With outside owners ($\theta=0$), a linear participation schedule gives the wage test a simple form. Let
\begin{equation}
    w(L)=w_0+\gamma L, \qquad w_0>0,\quad \gamma\geq0.
    \label{eq:linearwage}
\end{equation}
Condition~\eqref{ass:complementarity} becomes
\begin{equation}
    \kappa\mu\bigl[w_0+2\gamma\ell(1-d)\bigr]>\gamma d\ell.
    \label{eq:linearSC}
\end{equation}
With $\theta=0$ the left side of the complementarity condition, $\kappa\mu(w_0+2\gamma L)$, rises with $L$, while the right side, $d\ell\gamma$, is constant; this inequality is therefore necessary and sufficient for $G$ to increase over the entire feasible interval. For $\gamma>0$, define
\begin{equation}
    A=\kappa\mu\gamma,\qquad
    B=\kappa\mu w_0-d\gamma\ell,\qquad
    C=\kappa\overline M-d\ell w_0.
\end{equation}
In the coordination case the threshold is
\begin{equation}
    L^*=\frac{-B+\sqrt{B^2-4AC}}{2A},
    \qquad
    x^*=\frac{L^*/\ell-(1-d)}{d}.
    \label{eq:flexxstar}
\end{equation}

When $\gamma=0$, the local wage is pinned at $w_0$ by a perfectly elastic supply of specialized labor. Let $K=w_0\ell$ be the payroll of a full team and $q=M_{\HH}/K$ the ratio of a sector's sales to that payroll in the augmented allocation; with outside owners, $q=\overline M/K+\mu$. The displacement cutoffs are
\begin{equation}
    d_1=\frac{\kappa q\,(1-\mu\theta b_{\HH})}{1+\kappa\mu-\mu\theta b_{\HH}},
    \qquad
    d_2=\kappa q,
    \label{eq:rigidthresholds}
\end{equation}
with augmentation unique for $d<d_1$, coordination for $d_1<d<d_2$, and automation unique for $d>d_2$, within the parameter domain defined by Assumption~\ref{ass:viability} and $d\in(0,1]$. In particular, the automation-only range is empty when $d_2\geq1$. With outside owners, $d_1=\kappa q/(1+\kappa\mu)$ and $G$ is affine in $x$, which delivers the closed forms used in Section~\ref{sec:dynamics}.

\begin{proposition}[Ownership and the coordination interval]
\label{prop:ownership}
In the fixed-wage benchmark, hold $K$, $\mu$, $b_{\Atech}$, and $b_{\HH}$ fixed, and consider ownership shares for which $d_2\leq1$.
\begin{enumerate}
    \item At a fixed ratio $q$, the upper cutoff $d_2$ is independent of the resident ownership share $\theta$, while the lower cutoff $d_1$ strictly decreases in $\theta$. The distance between the cutoffs is
    \begin{equation}
        d_2-d_1=\frac{\kappa^2\mu q}{1+\kappa\mu-\mu\theta b_{\HH}},
        \label{eq:ownershipwidth}
    \end{equation}
    which strictly increases in $\theta$. When Assumption~\ref{ass:viability} holds at $d=d_1$, this distance is the width of the coordination interval.
    \item At fixed outside demand $\overline M$, the nonempty interval of displacement shares satisfying Assumption~\ref{ass:viability} and supporting coordination strictly widens with $\theta$.
\end{enumerate}
\end{proposition}

Resident ownership changes the channel through which deployment choices affect demand. With outside owners, the payroll that automation saves leaves the region as dividends, and the local market loses the spending of displaced workers. With resident owners, the saved payroll is spent locally, but each automated firm pays out a larger share of its sales for model usage, and the local multiplier $1/[1-\mu\theta b(x)]$ amplifies the effect of those payments leaving the region. In the fixed-output setting of \citet[Section~5.4]{falk2026layoff}, owners' spending narrows the over-automation wedge. Here, under the conditions of Proposition~\ref{prop:ownership}, resident ownership widens the coordination interval because payments for model usage leave the region whoever owns the firm. When the upper cutoff exceeds full displacement, the feasible interval ends at $d=1$. At fixed outside demand, resident ownership can then raise the lower cutoff above one, making augmentation the unique equilibrium throughout the viable displacement range.

Wage flexibility works in the opposite direction. Let the wage curve be
\begin{equation}
    w(L)=w_{\HH}\left(\frac{L}{\ell}\right)^{\eta},
    \qquad \eta\geq0,
    \label{eq:isoelasticwage}
\end{equation}
so that $\eta$ measures how strongly the wage responds to local employment and $\eta=0$ is the fixed-wage benchmark. With outside owners, $K=w_{\HH}\ell$, and the ratio $q$ held fixed,
\begin{equation}
    \frac{M(x)}{K}=q-\mu+\mu
    \left(\frac{L(x)}{\ell}\right)^{1+\eta},
    \label{eq:isodemand}
\end{equation}
and condition~\eqref{ass:complementarity} holds at a displacement share $d<1$ if
\begin{equation}
    \kappa\mu(1+\eta)(1-d)>d\eta.
    \label{eq:isocomplementarity}
\end{equation}

\begin{proposition}[Wage flexibility and the coordination interval]
\label{prop:wageflex}
Under equation~\eqref{eq:isoelasticwage} with outside owners, the upper cutoff is $d_2=\kappa q$ for every $\eta\geq0$. Suppose $d_2<1$ and condition~\eqref{eq:isocomplementarity} holds at $d=d_2$. Then $G_{\Atech}$ strictly decreases in $d$ on $(0,d_2]$, and whenever the coordination interval is nonempty its lower cutoff $d_1(\eta)$ is the unique root in $(0,d_2)$ of
\begin{equation}
    \kappa\left[q-\mu+\mu(1-d)^{1+\eta}\right]
    =d(1-d)^{\eta}.
    \label{eq:isocutoff}
\end{equation}
The lower cutoff strictly increases in $\eta$, so greater wage flexibility narrows the coordination interval.
\end{proposition}

A more responsive wage lowers the payroll that a deviating firm must carry in the automated allocation by more than it lowers the market that pays for that payroll, so the automation equilibrium needs more displacement to survive. The upper cutoff does not move because the wage in the augmented allocation is held at $w_{\HH}$.

\section{Reorganization and expectations}
\label{sec:dynamics}

Proposition~\ref{prop:flexclassification} identifies two stable allocations but not which one the economy reaches. Reorganization takes time: a firm that has automated cannot rebuild its team at once, and a firm with a team does not reorganize at every change in market conditions. We model this friction in the standard way \citep{calvo1983staggered,matsuyama1991increasing}. Each incumbent receives opportunities to change its mode, which we call revision opportunities, at random dates that arrive at rate $\rho>0$, independently across firms. In any short interval of length $dt$, a share $\rho\,dt$ of incumbents can reorganize, and the expected time between a firm's opportunities is $1/\rho$. Between opportunities, a firm's mode is fixed; we call a firm at a revision opportunity a revising firm.

\subsection{Myopic adjustment}

Suppose first that each firm, when it can reorganize, chooses the mode with the higher current profit. In the coordination case, Proposition~\ref{prop:flexclassification} gives
\begin{equation}
    \dot x=
    \begin{cases}
        -\rho x, & x<x^*,\\[2pt]
        \rho(1-x), & x>x^*.
    \end{cases}
    \label{eq:myopicdynamics}
\end{equation}
The inherited share selects the destination: below $x^*$ every revising firm automates, and above it every revising firm rebuilds its team. This benchmark produces a cascade in which beliefs play no role.

An improvement in the model moves this threshold and can start a transition. Holding the wage curve and the demand system fixed, an increase in the profitability of unattended automation raises $b_{\Atech}$ and lowers $\kappa$. With outside owners $M$ does not depend on $\kappa$, and the implicit-function theorem gives
\begin{equation}
    \frac{\partial x^*}{\partial\kappa}
    =-\frac{M(x^*)}{G'(x^*)}<0.
    \label{eq:kappastar}
\end{equation}
An improvement concentrated in mode $\Atech$ therefore raises $x^*$. Once the threshold passes the current augmentation share, later reorganizations carry the shock forward even after model quality stops changing.

\subsection{Forward-looking firms}

Firms now discount at rate $r>0$ and anticipate the path of the economy. A firm keeps its mode until its next opportunity, and therefore cares about the market over that whole spell. Switching from $\HH$ to $\Atech$ costs $F_{\Atech}\geq0$ once for integration, while rebuilding costs $F_{\HH}\geq0$ for recruitment, training, and the restoration of complementary domain expertise and organizational capital \citep{tambe2026reskilling}. These costs pay for real organizational inputs supplied from outside the region, financed by the firm's owners. With positive switching costs the analysis below takes ownership to be outside, $\theta=0$; the results with zero costs hold at every ownership share, and Section~\ref{app:ownplanner} of the e-companion treats positive costs under resident ownership. Declining a switch now postpones its cost to the next opportunity, so the cost of paying now rather than at that opportunity is
\begin{equation}
    \widehat F_j=\frac{r}{r+\rho}F_j.
    \label{eq:effectiveF}
\end{equation}

Consider two anticipated paths from an inherited share $x$. On the \emph{automation path} every future revising firm automates, and the share declines as $x_t=xe^{-\rho t}$; on the \emph{recovery path} every future revising firm augments, and the share rises as $x_t=1-(1-x)e^{-\rho t}$. Along each path, a firm evaluates the expected discounted gain from augmentation until its next opportunity:
\begin{align}
    \Phi_{\Atech}(x)
    &=\int_0^\infty e^{-(r+\rho)t}G(xe^{-\rho t})\,dt,
    \label{eq:PhiA}\\
    \Phi_{\HH}(x)
    &=\int_0^\infty e^{-(r+\rho)t}
      G\bigl(1-(1-x)e^{-\rho t}\bigr)\,dt.
    \label{eq:PhiH}
\end{align}
The factor $e^{-(r+\rho)t}$ combines discounting with the probability that the firm has had no further opportunity by date $t$. Both functions are strictly increasing under condition~\eqref{ass:complementarity}.

\begin{proposition}[Self-fulfilling transitions]
\label{prop:expectations}
Suppose the static economy is in the coordination case, with outside ownership when switching costs are positive.
\begin{enumerate}
    \item For an inherited share $x\in(0,1]$, the automation path $x_t=xe^{-\rho t}$ is a perfect-foresight equilibrium if and only if
    \begin{equation}
        \Phi_{\Atech}(x)\leq-\widehat F_{\Atech}.
        \label{eq:Apathcondition}
    \end{equation}
    At $x=0$, the constant path $x_t=0$ is always an equilibrium.
    \item For an inherited share $x\in[0,1)$, the recovery path $x_t=1-(1-x)e^{-\rho t}$ is a perfect-foresight equilibrium if and only if
    \begin{equation}
        \Phi_{\HH}(x)\geq\widehat F_{\HH}.
        \label{eq:Hpathcondition}
    \end{equation}
    At $x=1$, the constant path $x_t=1$ is always an equilibrium.
\end{enumerate}
With zero switching costs and interior dominance boundaries, there are unique thresholds $x_{\HH}^{E}<x^*<x_{\Atech}^{E}$ satisfying $\Phi_{\HH}(x_{\HH}^{E})=0$ and $\Phi_{\Atech}(x_{\Atech}^{E})=0$. Here interior dominance boundaries mean $\Phi_{\HH}(0)<0$ and $\Phi_{\Atech}(1)>0$, which place both roots strictly inside the feasible state space. Every inherited share in $[x_{\HH}^{E},x_{\Atech}^{E}]$ then supports both paths. The overlap remains nonempty for sufficiently small positive switching costs.
\end{proposition}

Starting above the static threshold, firms automate when they expect later reorganizations to shrink demand before their next opportunity; starting below it, they rebuild when they expect the market to grow. The two outcomes share the same fundamentals and the same inherited share and differ only in what firms expect others to do, the indeterminacy studied by \citet{krugman1991history} and \citet{matsuyama1991increasing}. The proposition uses only the monotonicity of $G$, so it holds with a market-clearing wage.

In the fixed-wage benchmark with outside owners, the gain is affine,
\begin{equation}
    G(x)=\Gamma(x-x^*),
    \qquad \Gamma=\kappa\mu dK>0,
    \label{eq:affineG}
\end{equation}
which is also the first-order approximation to any differentiable $G$ around its static threshold. With zero switching costs, the two thresholds of Proposition~\ref{prop:expectations} are
\begin{equation}
    x_{\HH}^{E}=x^*-\frac{\rho}{r+\rho}(1-x^*),
    \qquad
    x_{\Atech}^{E}=x^*+\frac{\rho}{r+\rho}x^*.
    \label{eq:zeroFthresholds}
\end{equation}
The interval between them collapses to $x^*$ as firms become myopic relative to the speed of adjustment, $r/\rho\to\infty$, and widens, up to the bounds of the state space, as firms give more weight to the market that later reorganizations will create. Switching costs narrow it because joining either transition requires an initial expenditure. Section~\ref{app:expectations} of the e-companion gives the interval with switching costs in closed form.

\paragraph{Equilibrium selection.} In the affine benchmark with zero switching costs, let a public profitability shock follow the driftless Brownian motion specified in Section~\ref{sec:selection} of the e-companion. The shock can reach values at which either mode dominates. This game satisfies the assumptions of the uniqueness theorem of \citet{frankel2000resolving}: almost every realization of the shock determines a unique equilibrium path, and history continues to matter when revision opportunities are infrequent. As revision opportunities become arbitrarily frequent, the selection boundary satisfies $\int_0^1G(x,z)\,dx=0$. For the affine gain, this condition selects the risk-dominant mode \citep{frankel2000resolving,burdzy2001fast}: at zero shock, augmentation if $x^*<1/2$ and automation if $x^*>1/2$. In the fixed-wage benchmark with outside owners, the threshold equals one half at the displacement share
\begin{equation}
    d_{\mathrm{RD}}=\frac{\kappa q}{1+\kappa\mu/2}.
    \label{eq:dRD}
\end{equation}
Since $d_1<d_{\mathrm{RD}}<d_2$, augmentation is selected within the coordination interval for $d<d_{\mathrm{RD}}$ and automation for $d>d_{\mathrm{RD}}$. An improvement in unattended automation lowers $\kappa$, and with it $d_{\mathrm{RD}}$ and $d_2$. As $d_{\mathrm{RD}}<d_2$, a given displacement share crosses $d_{\mathrm{RD}}$ first: the selected equilibrium switches to automation while the augmentation equilibrium still exists.

\section{Welfare and policy implementation}
\label{sec:welfare}

\subsection{Stationary welfare}

The welfare criterion is the joint surplus of the region's workers and of the owners of its service firms, wherever those owners live; it excludes the model supplier. Limit pricing leaves buyers paying the generic service's price for each unit of quality whatever the incumbent's mode, so buyers' surplus does not vary with the allocation. Up to terms independent of $x$, joint surplus is
\begin{equation}
    S(x)=b(x)M(x)-C(L(x)).
    \label{eq:welfare}
\end{equation}
Wage payments cancel as transfers between owners and workers, $C(L)$ records the opportunity cost of drawing specialized workers into the local sector, and payments for model usage are outflows. We call the allocation that maximizes $S$ the local first best.

The planner's marginal return to an increase in augmentation is
\begin{equation}
    S'(x)=\kappa M(x)+b(x)M'(x)-d\ell w(L(x)).
    \label{eq:welfareprime}
\end{equation}
Comparing equation~\eqref{eq:welfareprime} with the private gain in equation~\eqref{eq:gain2} identifies the adoption externality:
\begin{equation}
    \mathcal E(x)=S'(x)-G(x)=b(x)M'(x)>0.
    \label{eq:externality}
\end{equation}
An additional augmenting firm raises spending on every sector by $M'(x)$, and each dollar of that spending yields $b(x)$ dollars of local wages and profit. The firm's own profit records none of it: as in the factory-wage model of \citet{murphy1989industrialization}, profit understates a firm's contribution to demand for other sectors.

Assume in this section that the wage bill $W(L)=w(L)L$ is convex, as it is under the linear and constant-elasticity schedules used above. Since $L(x)$ is affine, the numerator of equation~\eqref{eq:marketx} is convex in $x$, while its denominator is positive and falls linearly; hence $M$ is convex, $M'$ is nondecreasing, and $M'>0$ by equation~\eqref{eq:Mprime}. Using equation~\eqref{eq:externality}, write
\begin{equation}
    S'(x)=G(x)+b(x)M'(x).
    \label{eq:welfareconvex}
\end{equation}
The first term is strictly increasing, and the second is nondecreasing because $b$ is positive and increasing and $M'$ is positive and nondecreasing. Thus $S'$ is strictly increasing and joint surplus is strictly convex in the augmentation share, so the local first best is one of the two uniform allocations.

\begin{proposition}[Local first best]
\label{prop:firstbest}
Suppose Assumption~\ref{ass:quality} and condition~\eqref{ass:complementarity} hold and $W$ is convex. In the coordination case, the augmented allocation is the unique local first best:
\begin{equation}
    S(1)>S(0).
    \label{eq:welfareranking}
\end{equation}
This ranking holds for every ownership share.
\end{proposition}

The proof needs no condition on the wage response. In the coordination case $G_{\HH}>0$, that is, $\kappa M_{\HH}>d\ell w_{\HH}$. Since the wage is nondecreasing,
\begin{equation}
    C(\ell)-C(\ell(1-d))\leq d\ell w_{\HH}.
\end{equation}
The private gain from augmentation in the augmented allocation therefore already covers the opportunity cost of the additional workers, and the local value added generated by the larger market strengthens the inequality.

Counting only the region's residents leads to the same allocation. Residents' welfare consists of workers' participation rents and the resident share of profits,
\begin{equation}
    R(x)=W(x)-C(L(x))+\theta\Pi(x)=(1-\theta)\bigl[W(x)-C(L(x))\bigr]+\theta S(x).
    \label{eq:residentwelfare}
\end{equation}

\begin{corollary}[Resident welfare]
\label{cor:resident}
Under the conditions of Proposition~\ref{prop:firstbest}, the augmented allocation maximizes residents' welfare. It is the unique maximizer when $\theta>0$ or the wage rises strictly with employment.
\end{corollary}

\subsection{Optimal reorganization}
\label{sec:dynplanner}

The augmented allocation is the local first best, but reaching it from an inherited share $x<1$ requires rebuilding teams, and a planner must weigh the surplus generated along a transition against the resources spent on reorganization. Joint surplus is strictly convex in the augmentation share, so the planner's problem is not concave: first-order conditions do not identify its solution, and plans that adjust partially, wait, or change direction have to be compared directly. We solve the problem globally. Begin with outside ownership, $\theta=0$, and retain the market-clearing wage. Maintain the conditions of Proposition~\ref{prop:firstbest} and take $S$ to be twice continuously differentiable, so that $S$ is strictly convex and $S(1)>S(0)$.

At a revision opportunity, an automated firm switches to augmentation with probability $u_{\HH}\in[0,1]$, and an augmenting firm switches to automation with probability $u_{\Atech}\in[0,1]$. The planner chooses these probabilities over time. The augmentation share and expenditure on organizational inputs obey
\begin{align}
    \dot x=f(x,u)&=\rho[(1-x)u_{\HH}-xu_{\Atech}],
    \label{eq:dyndrift}\\
    \mathcal K(x,u)&=\rho[(1-x)u_{\HH}F_{\HH}+xu_{\Atech}F_{\Atech}].
    \label{eq:dyncost}
\end{align}
The switching costs defined in Section~\ref{sec:dynamics} pay for inputs supplied outside the region. With outside owners, expenditure on these inputs reduces net dividends by the same amount, leaving resident income and $M(x)$ unchanged. Joint flow surplus is $S(x)-\mathcal K(x,u)$. Firms can finance reorganization when an opportunity arrives.

The planner discounts at the same rate as firms and chooses measurable controls to maximize
\begin{equation}
    V(x)=\sup_{u_{\HH},u_{\Atech}}\int_0^\infty
    e^{-rt}[S(x_t)-\mathcal K(x_t,u_t)]\,dt,
    \qquad x_0=x.
    \label{eq:dynproblem}
\end{equation}
Equation~\eqref{eq:dyndrift} respects the state bounds for every admissible control. At a differentiability point, the Hamilton--Jacobi--Bellman equation is
\begin{equation}
    rV=S+\rho(1-x)[V'-F_{\HH}]_+
          +\rho x[-V'-F_{\Atech}]_+,
    \qquad [z]_+\equiv\max\{z,0\}.
    \label{eq:dynHJB}
\end{equation}
The planner rebuilds when the lifetime value of an additional team exceeds its cost and automates when the value of reducing the augmentation share covers the cost of that switch.

Three policies have simple discounted values. Inaction preserves the inherited share forever. Recovery uses every revision opportunity to rebuild and follows the recovery path $x_t^{\HH}=1-(1-x)e^{-\rho t}$; automation uses every opportunity to reduce employment and follows the automation path $x_t^{\Atech}=xe^{-\rho t}$. Their values, including reorganization expenditure, are
\begin{align}
    N(x)&=\frac{S(x)}r,
    \label{eq:dynN}\\
    H(x)&=\int_0^\infty e^{-rt}S(x_t^{\HH})\,dt
          -\frac{\rho(1-x)F_{\HH}}{r+\rho},
    \label{eq:dynH}\\
    A(x)&=\int_0^\infty e^{-rt}S(x_t^{\Atech})\,dt
          -\frac{\rho xF_{\Atech}}{r+\rho}.
    \label{eq:dynA}
\end{align}

\begin{proposition}[Optimal reorganization]
\label{prop:dynplanner}
Under the conditions above, the planner's value over all admissible reorganization plans is
\begin{equation}
    V(x)=\max\{N(x),H(x),A(x)\}.
    \label{eq:dynenvelope}
\end{equation}
There is an optimal policy that starts one of these three paths immediately. The state space consists, in order, of an automation region, an inaction interval, and a recovery region, any of which may be empty. Once adjustment begins, it uses every available opportunity in the chosen direction. At a boundary where inaction and adjustment have equal values, delayed adjustment can also be optimal.
\end{proposition}

Recovery expands the market that makes subsequent rebuilding valuable, while automation reduces the return to the teams that remain. These increasing returns make the relative value of recovery rise with the inherited augmentation share. The proof verifies that the largest of the three candidate values satisfies the planner's optimality conditions at every state, giving a bound on the value of every feasible plan. Each candidate attains that bound where it is largest, so no plan with partial switching, waiting, or reversals does better.

The fixed-wage benchmark yields explicit policy boundaries. With $K=w_0\ell$ as above, write
\begin{equation}
    S(x)=S(0)+\Omega\left(\frac{x^2}{2}-x^Sx\right),
    \qquad \Omega=2\kappa\mu dK>0.
    \label{eq:dynquadratic}
\end{equation}
Here $x^S$ is the algebraic zero of $S'$ and can lie outside $[0,1]$. Proposition~\ref{prop:firstbest} implies $x^S<1/2$. Define
\begin{equation}
    a=\frac{r+2\rho}{r+\rho}
      \left(x^S-\frac{rF_{\Atech}}{\Omega}\right),
    \qquad
    h=\frac{(r+2\rho)(x^S+rF_{\HH}/\Omega)-\rho}{r+\rho}.
    \label{eq:dynboundaries}
\end{equation}
For an interior state, automation is worth at least as much as inaction exactly when $x\leq a$, while recovery is worth at least as much as inaction exactly when $x\geq h$.

\begin{corollary}[Adjustment and inaction]
\label{cor:dynregions}
In the fixed-wage, outside-ownership benchmark, if
\begin{equation}
    F_{\Atech}+F_{\HH}\geq\frac{\Omega\rho}{r(r+2\rho)},
    \label{eq:dyncostthreshold}
\end{equation}
then $a\leq h$. The planner automates below $a$, preserves the inherited share between $a$ and $h$, and recovers above $h$, intersecting these regions with the feasible state space.

If the inequality is reversed, $h<a$ and there is no interior inaction region. Automation is optimal for $x<c$ and recovery for $x>c$, where
\begin{equation}
    c=\frac{(r+2\rho)x^S-\rho+r(r+2\rho)F_{\HH}/\Omega}
    {r[1+(r+2\rho)(F_{\Atech}+F_{\HH})/\Omega]}.
    \label{eq:dynswitch}
\end{equation}
Both adjustment paths are optimal at an interior $c$. At the stationary endpoints, equation~\eqref{eq:dynenvelope} gives the relevant decision.
\end{corollary}

A higher rebuilding cost raises the augmentation share required to justify recovery, and a higher cost of automating lowers the share below which automation is optimal. When the two costs are large enough together, an interval of inherited shares is worth preserving: from those shares the gain from reaching the augmented allocation falls short of the discounted cost of rebuilding. When adjustment is cheaper, the planner chooses a direction at once.

Recovery begins before its current marginal surplus turns positive. With zero switching costs, the boundary in equation~\eqref{eq:dynswitch} satisfies
\begin{equation}
    c_0=\frac{(r+2\rho)x^S-\rho}{r},
    \qquad
    c_0-x^S=\frac{\rho}{r}(2x^S-1)<0.
    \label{eq:dynzerocost}
\end{equation}
The planner therefore rebuilds from shares between $c_0$ and $x^S$, at which an additional team lowers current surplus, because each rebuilt team enlarges the market for those rebuilt later. Greater patience or faster reorganization lowers the boundary further.

\begin{samepage}
With positive costs, recovery is optimal from every inherited share, allowing indifference, exactly when
\begin{equation}
    (r+2\rho)\left(x^S+\frac{rF_{\HH}}{\Omega}\right)\leq\rho.
    \label{eq:dynrecoveryall}
\end{equation}
\end{samepage}
This condition compares rebuilding from the automated allocation with remaining there. Since recovery's relative value rises with the augmentation share, the comparison at that endpoint determines whether recovery is optimal everywhere.

The characterization extends to resident ownership. If residents own a share $\theta$ of firms, expenditure $\mathcal K$ reduces their profit income and the demand it supports, so the social flow cost of reorganization is $\mathcal K/[1-\mu\theta b(x)]$. Section~\ref{app:ownplanner} of the e-companion shows (Proposition~\ref{prop:ownplanner}) that, under a condition on the curvature of surplus, the planner's value is again the largest of the three policy values with these state-dependent costs, and it derives the wedge that supports the optimal path.

\subsection{Policy implementation}
\label{sec:policy}

Let $s(x)$ denote a \emph{wedge}: a payment to mode $\HH$ relative to mode $\Atech$ at augmentation share $x$. The Pigouvian wedge equals the adoption externality,
\begin{equation}
    s^{P}(x)=\mathcal E(x)=b(x)M'(x).
    \label{eq:pigou}
\end{equation}
It is the unique wedge that makes the private flow gain $G+s^P$ equal the planner's marginal return $S'$.

A \emph{balanced transfer} among firms implements any wedge $s(x)$ without public funds and without changing local spending. Each automated firm pays $x\,s(x)$, and each augmenting firm receives $(1-x)\,s(x)$. The payments sum to zero, so aggregate profit, residents' income, the market size $M(x)$, and the surplus $S(x)$ are unchanged for every ownership share, while the gain from augmentation rises by $(1-x)s+xs=s$. In either uniform allocation no firm pays or receives anything, so the transfer collects nothing once every firm has a team. We assume that the wedge leaves automation viable net of the payment, $b_{\Atech}M(x)-(1-d)\ell w(L(x))\geq x\,s(x)$, the counterpart of Assumption~\ref{ass:viability}, so that a revising firm never prefers to cede its sector to the fringe. Every wedge below is implemented in this way.

Regulators already use transfers of this form to promote electric vehicles. California's Zero-Emission Vehicle regulation has required automakers to hold credits in proportion to their sales in the state \citep{carb2022zev}, and China's dual-credit policy has placed a similar requirement on large producers of conventional passenger cars since 2019 \citep{miit2017dualcredit}. Selling electric vehicles earns credits, and an automaker that falls short covers the gap with credits bought from automakers that have a surplus. The regulator sets the requirement, and the payments pass from firm to firm: automakers in China paid one another 10.9 billion yuan for credits in 2021 \citep{li2022dualcredit}, and Tesla sold \$6.5 billion of regulatory credits to other automakers under these and similar programs from 2023 to 2025 \citep{tesla2026annual}. Our balanced transfer has the same structure, with automated firms in the role of automakers that fall short and augmenting firms in the role of automakers with surplus credits.

The Pigouvian wedge also aligns lifetime incentives when firms pay switching costs. Continue with the outside-ownership planner in Section~\ref{sec:dynplanner}. Along an anticipated aggregate path, let $J_{\HH}$ and $J_{\Atech}$ be a firm's values in the two modes and let $D=J_{\HH}-J_{\Atech}$. An automated firm rebuilds when $D>F_{\HH}$, and an augmenting firm automates when $D<-F_{\Atech}$. Between revision opportunities, the value gap obeys
\begin{equation}
    rD=G+s+\dot D
       +\rho[-D-F_{\Atech}]_+
       -\rho[D-F_{\HH}]_+,
    \label{eq:dynprivate}
\end{equation}
where $\dot D$ is the change in the gap along the anticipated path. The last two terms account for the option to change mode at a future opportunity.

\begin{proposition}[Dynamic implementation]
\label{prop:dynimplementation}
Under the conditions of Proposition~\ref{prop:dynplanner} and the maintained financing and operating-viability conditions, the Pigouvian wedge $s^P(x)=b(x)M'(x)$, implemented by the balanced transfer, supports each planner-optimal path as a decentralized equilibrium, with the corresponding expectations and choices at indifference.
\end{proposition}

On a smooth optimal path, the planner's marginal value of an additional team obeys the same equation as the private value gap, with $S'$ in place of $G+s$. The Pigouvian wedge makes those flow returns equal. Private and social lifetime values then coincide, so firms take the switching decisions that attain the planner's value.

The planner's solution also ranks paths when the Pigouvian wedge supports more than one equilibrium. In the fixed-wage benchmark, the wedge supports the recovery path at interior shares $x\geq h$ and the automation path at $x\leq a$. When $h<a$, both are equilibria on their feasible overlap, and the threshold $c$ in Corollary~\ref{cor:dynregions} gives the one with the higher value: recovery above $c$ and automation below it. Coordinating expectations on that path implements the optimum.

In the coordination case the augmentation equilibrium sustains itself once reached, and the balanced transfer collects nothing there. Starting a recovery from few surviving teams can also call for \emph{temporary support}, a payment added to the Pigouvian wedge while the recovery gets under way. For outside ownership, Section~\ref{app:withdrawal} of the e-companion gives a threshold beyond which temporary support can be withdrawn while the Pigouvian wedge stays in force, provided $rF_{\HH}<S'(1)$. The same separation between a maintained wedge and temporary support arises in \citet{alvarez2026technology}. When firms respond to current payoffs, Section~\ref{app:myopicsupport} characterizes the smallest temporary support. It also derives the \emph{selection wedge}, the smallest wedge that moves the economy to the augmentation equilibrium if the Pigouvian wedge is not maintained. The selection wedge ends once the augmentation share passes the static threshold.

\section{Extensions}
\label{sec:extensions}

\subsection{Fixed costs of model use}

Model suppliers also sell subscriptions, dedicated capacity, and integration services. Suppose an automated firm pays the supplier a fixed amount $f_{\Atech}\geq0$ per period in addition to its usage costs. The fee lowers the profits of automated firms, and residents spend their share of profits locally, so the market size becomes
\begin{equation}
    M_f(x)=\frac{\overline M+\mu(1-\theta)W(x)-\mu\theta(1-x)f_{\Atech}}{1-\mu\theta b(x)},
    \label{eq:feemarket}
\end{equation}
and the gain from augmentation becomes
\begin{equation}
    G_f(x)=\kappa M_f(x)-\bigl[d\ell w(L(x))-f_{\Atech}\bigr].
    \label{eq:fixedfee}
\end{equation}
With outside owners, $M_f=M$ and the fee shifts the gain without changing its slope. With resident owners, a firm that switches to augmentation also stops paying a fee that was financed partly out of local spending, so
\begin{equation}
    M_f'(x)=M'(x)+\frac{\mu\theta f_{\Atech}(1-\mu\theta b_{\HH})}{[1-\mu\theta b(x)]^2}\geq M'(x),
    \label{eq:feeslope}
\end{equation}
and the fee strengthens the complementarity. The fee is a payment leaving the region, so joint surplus becomes $b(x)M_f(x)-C(L(x))-(1-x)f_{\Atech}$, and the adoption externality becomes $b(x)M_f'(x)$. With Assumption~\ref{ass:viability} holding net of the fee, the equilibrium classification, path conditions of Proposition~\ref{prop:expectations}, and stationary welfare result of Proposition~\ref{prop:firstbest} apply with $M$ replaced by $M_f$ and the payroll saving $d\ell w(L)$ replaced by the net saving $d\ell w(L)-f_{\Atech}$. Defensive automation requires only that automation lower the firm's fixed costs, $f_{\Atech}<d\ell w(L)$.

\subsection{Re-employment}

Let $\chi\in[0,1]$ be the share of earnings lost through automation that displaced workers replace quickly enough to support contemporaneous demand for local services. With outside owners, and holding the market size in the augmented allocation fixed, equation~\eqref{eq:marketx} becomes
\begin{equation}
    M_{\chi}(x)=\overline M+\mu\left[(1-\chi)W(x)+\chi W_{\HH}\right].
    \label{eq:reemployment2}
\end{equation}
The gap between the two market sizes is
\begin{equation}
    M_{\chi}(1)-M_{\chi}(0)
    =(1-\chi)\mu\bigl(W_{\HH}-W_{\Atech}\bigr).
\end{equation}
Re-employment leaves the augmented allocation unchanged and narrows this gap, which remains positive for every $\chi<1$. The coordination interval contracts as $\chi$ rises. The demand channel runs through earnings that cease to support local spending. Replacement earnings sustain the market when workers continue spending them in the region, including earnings from jobs outside it. When re-employment moves the associated spending elsewhere, it leaves the demand loss unchanged.

\subsection{Offensive and defensive automation}

In offensive service lines, automation earns at least as much as augmentation from a larger market, and the single-crossing condition fails. Divide the economy into defensive sectors $D$ and offensive sectors $O$, with masses $\omega$ and $1-\omega$. Employment is
\begin{equation}
    L=\omega\ell_D(1-d_D+d_Dx_D)
      +(1-\omega)\ell_O(1-d_O+d_Ox_O),
    \label{eq:twosectoremployment}
\end{equation}
and, with outside owners, both groups face $M=\overline M+\mu w(L)L$. Let $\kappa_D>0$, while $\kappa_O\leq0$ captures an offensive application in which unattended automation has at least as large an incremental return to the market as augmentation. The gain in a defensive sector is
\begin{equation}
    G_D=\kappa_D M-d_D\ell_Dw(L).
\end{equation}
Its response to augmentation in the offensive sector is
\begin{equation}
    \frac{\partial G_D}{\partial x_O}
    =(1-\omega)d_O\ell_O
    \left\{\kappa_D\mu[w(L)+Lw'(L)]
    -d_D\ell_Dw'(L)\right\}.
    \label{eq:crosssector}
\end{equation}
Whenever the bracket is positive, offensive automation lowers the profitability of augmentation in defensive sectors through the wage income spent in the common local market. A capability improvement that begins where automation expands productive opportunity can in this way continue as a defensive cascade in sectors where the direct technological case is weaker.

\subsection{Heterogeneous firms}

Consider outside ownership and let firm $i$ pay an additional flow cost $\varepsilon_i\geq0$ for augmentation inputs supplied outside the region. These expenses reduce outside owners' dividends and leave resident income, market size, and wages unchanged at a given augmentation share. Let the costs have a continuously differentiable distribution function $F$ and density $f$. The firm's mode gain is $G(x)-\varepsilon_i$, so the augmentation share satisfies
\begin{equation}
    x=F(G(x)).
    \label{eq:heterogeneous}
\end{equation}
Condition~\eqref{ass:complementarity} makes the right side increasing. The slope of the fixed-point map is $f(G(x))G'(x)$, so a necessary condition for multiple interior fixed points is
\[
    \sup_{x\in[0,1]} f(G(x))G'(x)\geq1.
\]
Multiple isolated transversal fixed points require the slope to exceed one somewhere. When it does so over an intermediate range and a location shifter generates the required boundary crossings, the model can have three fixed points. Heterogeneity replaces the common threshold with a distribution of thresholds, and multiplicity is governed by the product of the density of firms near indifference and the slope of the demand feedback.

\section{Conclusion}

A firm's deployment of generative AI changes the market supporting other firms' human teams. Replacing its own team with purchased model capability lowers its break-even scale and, through the income that leaves the region, the sales of every other local firm. We have shown that the second effect can make the first self-justifying. Over a range of technologies, retaining teams and automating are both stable equilibria. Every firm earns more when all retain their teams, yet each would lose by rebuilding alone in the automated economy. The augmented allocation also yields a higher joint surplus. Both equilibria can persist with market-clearing wages and resident ownership.

The reverse big push changes what a surviving workforce is worth. A team contributes its own margin, the demand that supports other firms' teams, and the avoided cost of rebuilding it later, and the planner's solution accounts for all three. It gives the share of surviving teams above which recovery repays its cost, the share below which continued automation is optimal, and the interval between them in which the inherited workforce should be preserved. A budget-balanced transfer among firms, of the kind California and China use to promote electric vehicles, supports each optimal path as an equilibrium. It uses no public funds, and no firm pays once teams are restored. When several equilibrium paths remain, implementation also requires coordinating expectations on the planner's preferred path.

Managers therefore need to evaluate payroll commitments against customers' exposure to local earnings, other firms' deployment plans, and the cost of restoring a team once lost. A contraction in demand can trigger automation with unchanged model capability, and deployment changes in unrelated sectors can propagate through employees' spending. An improvement in unattended capability can tip a region toward automation while augmentation is still an equilibrium, and later reorganization can continue the cascade after capability stops improving. Firms serving the same local market consequently have a stake in one another's human teams.

\begingroup
\parindent 0pt
\parskip 0pt
\renewcommand{\enotesize}{\normalsize}
\theendnotes
\endgroup

\bibliographystyle{plainnat}
\bibliography{reverse_big_push_references}

\clearpage

\setcounter{section}{0}
\setcounter{equation}{0}
\setcounter{lemma}{0}
\setcounter{proposition}{0}
\setcounter{corollary}{0}
\setcounter{assumption}{0}
\setcounter{definition}{0}
\renewcommand{\thesection}{EC.\arabic{section}}
\renewcommand{\theequation}{EC.\arabic{equation}}
\renewcommand{\thelemma}{EC.\arabic{lemma}}
\renewcommand{\theproposition}{EC.\arabic{proposition}}
\renewcommand{\thecorollary}{EC.\arabic{corollary}}
\renewcommand{\theassumption}{EC.\arabic{assumption}}
\renewcommand{\thedefinition}{EC.\arabic{definition}}

\renewcommand{\theHsection}{\thesection}
\renewcommand{\theHequation}{\theequation}
\def\theHlemma{\thelemma}
\def\theHproposition{\theproposition}
\def\theHcorollary{\thecorollary}
\def\theHassumption{\theassumption}
\def\theHdefinition{\thedefinition}

\section*{E-Companion: Proofs and Additional Results}

\section{Proofs}

\begin{proof}[Proof of Proposition~\ref{prop:flexclassification}]
Equation~\eqref{eq:Gprime} and condition~\eqref{ass:complementarity} imply that $G$ is strictly increasing on $[0,1]$. If $G(0)>0$, mode $\HH$ is strictly optimal at every attainable state and $x=1$ is the unique equilibrium. If $G(1)<0$, mode $\Atech$ is strictly optimal at every state and $x=0$ is unique. If $G(0)<0<G(1)$, both endpoints satisfy individual optimality. Continuity and strict monotonicity give a unique $x^*\in(0,1)$ with $G(x^*)=0$; a measure $x^*$ of indifferent firms can support that allocation. Under myopic Poisson revision, $G(x)<0$ below $x^*$ and $G(x)>0$ above it, which yields equation~\eqref{eq:myopicdynamics} and the stated stability properties.
\end{proof}

\begin{proof}[Derivation of equation~\eqref{eq:flexxstar}]
Under equation~\eqref{eq:linearwage} and with outside owners, write the gain as a function of employment:
\begin{equation*}
    G(L)=\kappa\mu\gamma L^2
    +(\kappa\mu w_0-d\gamma\ell)L
    +\kappa\overline M-d\ell w_0.
\end{equation*}
The relevant quadratic root is the expression for $L^*$ in equation~\eqref{eq:flexxstar}; condition~\eqref{eq:linearSC} makes it the unique root in $[\ell(1-d),\ell]$. Inverting equation~\eqref{eq:employment} gives the expression for $x^*$.
\end{proof}

\begin{proof}[Proof of Proposition~\ref{prop:firmranking}]
Since $M_{\Atech}>0$, the inequality $G_{\Atech}<0$ implies $w_{\Atech}>0$. Viability at $x=0$ and $G_{\Atech}<0$ give
\[
    \kappa(1-d)\ell w_{\Atech}
    \leq \kappa b_{\Atech}M_{\Atech}
    < b_{\Atech}d\ell w_{\Atech},
\]
so $b_{\Atech}d/\kappa>1-d$. Using
$G_{\HH}-G_{\Atech}=\kappa\Delta M-d\ell\Delta w$,
\begin{align*}
    &b_{\Atech}\Delta M-(1-d)\ell\Delta w\\
    &\quad=
    \frac{b_{\Atech}}{\kappa}(G_{\HH}-G_{\Atech})
    +\left[\frac{b_{\Atech}d}{\kappa}-(1-d)\right]\ell\Delta w
    >0,
\end{align*}
because $G_{\HH}-G_{\Atech}>0$ and $\Delta w\geq0$.
This proves equation~\eqref{ass:ranking}. Adding $G_{\HH}>0$
to this difference gives equation~\eqref{eq:flexprofitranking}.
\end{proof}

\begin{proof}[Derivation of equation~\eqref{eq:rigidthresholds} and proof of Proposition~\ref{prop:ownership}]
With $w=w_0$ and $K=w_0\ell$, write $\delta=\mu\theta b_{\HH}$. Equation~\eqref{eq:marketx} and the definition of $q$ give $\overline M=qK(1-\delta)-\mu(1-\theta)K$, so
\begin{equation*}
    M_{\HH}=qK,
    \qquad
    M_{\Atech}=K\,\frac{q(1-\delta)-\mu(1-\theta)d}{1-\mu\theta b_{\Atech}}.
\end{equation*}
The boundary $G_{\HH}=0$ gives $d_2=\kappa q$. The boundary $G_{\Atech}=0$ gives $\kappa[q(1-\delta)-\mu(1-\theta)d]=d(1-\mu\theta b_{\Atech})$. Since $1-\mu\theta b_{\Atech}+\kappa\mu(1-\theta)=1+\kappa\mu-\delta$, its solution is $d_1=\kappa q(1-\delta)/(1+\kappa\mu-\delta)$. At fixed $q$,
\begin{equation*}
    \frac{\partial}{\partial\delta}\,\frac{1-\delta}{1+\kappa\mu-\delta}=-\frac{\kappa\mu}{(1+\kappa\mu-\delta)^2}<0
\end{equation*}
Thus $d_1$ strictly decreases in $\theta$, while $d_2$ is constant. Subtracting the cutoffs gives equation~\eqref{eq:ownershipwidth}. With a fixed wage, $G'=\kappa M'>0$ by equation~\eqref{eq:Mprime}, so Proposition~\ref{prop:flexclassification} classifies the three ranges.

Since $M(x)$ increases in $x$, automated profit is lowest at $x=0$. At that allocation,
\begin{equation*}
    \frac{\profit_{\Atech}(M_{\Atech},w_0)}{K}
    =
    \frac{b_{\Atech}\overline M/K-(1-\mu b_{\Atech})(1-d)}
    {1-\mu\theta b_{\Atech}}.
\end{equation*}
This expression increases in $d$, so viability at $d=d_1$ ensures viability throughout the coordination interval and completes part 1.

For part 2,
\begin{equation*}
    q(\theta)=\frac{\overline M/K+\mu(1-\theta)}
    {1-\mu\theta b_{\HH}},
    \qquad
    q'(\theta)=\frac{\mu[b_{\HH}q(\theta)-1]}
    {1-\mu\theta b_{\HH}}.
\end{equation*}
A nonempty coordination interval satisfying viability implies $\profit_{\HH}(M_{\HH},w_0)>\profit_{\Atech}(M_{\HH},w_0)\geq0$. Hence $b_{\HH}q(\theta)>1$ and $q'(\theta)>0$. Differentiating the logarithm of equation~\eqref{eq:ownershipwidth} gives
\begin{equation*}
    \frac{d\ln(d_2-d_1)}{d\theta}
    =\frac{\mu\,[b_{\HH}q(\theta)-1]}{q(\theta)(1-\mu\theta b_{\HH})}
    +\frac{\mu b_{\HH}}{1+\kappa\mu-\mu\theta b_{\HH}}>0.
\end{equation*}

The automated-profit expression also shows that viability is equivalent to $d\geq d_V$, where
\begin{equation*}
    d_V=1-\frac{b_{\Atech}\overline M}
    {K(1-\mu b_{\Atech})}.
\end{equation*}
At fixed outside demand, $d_V$ is independent of $\theta$. Since $d_2\leq1$, the nonempty coordination interval satisfying viability has width
\begin{equation*}
    d_2-\max\{d_1,d_V\}
    =\min\{d_2-d_1,\ d_2-d_V\}.
\end{equation*}
The first expression inside the minimum strictly increases by the preceding derivative. The second strictly increases because $d_2'=\kappa q'>0$ and $d_V$ is constant. Their minimum therefore strictly increases.
\end{proof}

\begin{proof}[Proof of Proposition~\ref{prop:wageflex}]
With outside owners and $K=w_{\HH}\ell$, $W(x)=K(L/\ell)^{1+\eta}$ and $\overline M=(q-\mu)K$, which give equation~\eqref{eq:isodemand}. At $x=1$, $G_{\HH}=K(\kappa q-d)$, so $d_2=\kappa q$. At $x=0$, $G_{\Atech}=K\,H(d,\eta)$ with
\begin{equation*}
    H(d,\eta)=\kappa(q-\mu)+(1-d)^{\eta}\bigl[\kappa\mu(1-d)-d\bigr],
\end{equation*}
the difference between the two sides of equation~\eqref{eq:isocutoff}. Differentiating,
\begin{equation*}
    \frac{\partial H}{\partial d}
    =-(1-d)^{\eta-1}\bigl[\kappa\mu(1+\eta)(1-d)+(1-d)-\eta d\bigr].
\end{equation*}
The left side of condition~\eqref{eq:isocomplementarity} falls with $d$ and its right side rises, so the condition holds on $(0,d_2]$ once it holds at $d_2$, and there $\partial H/\partial d<0$. Since $H(0,\eta)=\kappa q>0$, a nonempty coordination interval has a unique lower cutoff $d_1(\eta)\in(0,d_2)$ with $H(d_1,\eta)=0$. As $q>\mu$, that equation requires $\kappa\mu(1-d_1)<d_1$, so
\begin{equation*}
    \frac{\partial H}{\partial\eta}
    =(1-d)^{\eta}\ln(1-d)\bigl[\kappa\mu(1-d)-d\bigr]>0
    \quad\text{at }d=d_1(\eta).
\end{equation*}
The implicit-function theorem gives $d_1'(\eta)=-(\partial H/\partial\eta)/(\partial H/\partial d)>0$.
\end{proof}

\begin{proof}[Proof of Proposition~\ref{prop:expectations}]
Consider first an inherited share $x>0$ and the automation path. Let $T\sim\operatorname{Exp}(\rho)$ be the time until a firm's next revision. The expected discounted flow advantage of retaining $\HH$ until $T$ is
\begin{align*}
    \E\left[\int_0^T e^{-rt}G(xe^{-\rho t})\,dt\right]
    &=\int_0^\infty e^{-(r+\rho)t}G(xe^{-\rho t})\,dt\\
    &=\Phi_{\Atech}(x).
\end{align*}
Switching immediately costs $F_{\Atech}$; waiting entails the discounted expected cost $\E[e^{-rT}]F_{\Atech}=\rho F_{\Atech}/(r+\rho)$. Thus retaining $\HH$ rather than switching immediately is worth $\Phi_{\Atech}(x)+\widehat F_{\Atech}$. A positive mass of firms is initially in mode $\HH$ when $x>0$, so immediate automation by all revising $\HH$ firms is optimal if and only if equation~\eqref{eq:Apathcondition} holds. Under that condition, an incumbent already in $\Atech$ cannot gain by switching temporarily to $\HH$: the deviation is worth
\begin{align*}
    \Phi_{\Atech}(x)-F_{\HH}-\frac{\rho}{r+\rho}F_{\Atech}
    &\leq-\widehat F_{\Atech}-F_{\HH}
      -\frac{\rho}{r+\rho}F_{\Atech}\\
    &=-F_{\HH}-F_{\Atech}\leq0.
\end{align*}
Monotonicity of $\Phi_{\Atech}$ implies that the automation condition continues to hold as $x_t$ falls.

At $x=0$, there are no $\HH$ incumbents whose participation in the automation path requires equation~\eqref{eq:Apathcondition}. An $\Atech$ incumbent's gain from switching temporarily to $\HH$ is
\[
    \Phi_{\Atech}(0)-F_{\HH}
    -\frac{\rho}{r+\rho}F_{\Atech}<0,
\]
because $G(0)<0$ in the coordination case; any longer spell in mode $\HH$ only adds negative flow gains. Hence the constant path $x_t=0$ is always an equilibrium.

The argument for the recovery path is symmetric. For $x<1$, an $\Atech$ incumbent's gain from switching immediately to $\HH$ rather than waiting until its next revision is $\Phi_{\HH}(x)-\widehat F_{\HH}$. Thus all revising $\Atech$ firms choose $\HH$ if and only if equation~\eqref{eq:Hpathcondition} holds. Under that condition, an $\HH$ incumbent's gain from switching temporarily to $\Atech$ is at most
\[
    -\widehat F_{\HH}-F_{\Atech}
    -\frac{\rho}{r+\rho}F_{\HH}
    =-F_{\Atech}-F_{\HH}\leq0.
\]
The condition remains satisfied as $x_t$ rises. At $x=1$, there are no $\Atech$ incumbents who must be induced to rebuild, while $G(1)>0$ makes every spell in mode $\Atech$ unprofitable for an $\HH$ incumbent. The constant path $x_t=1$ is therefore always an equilibrium.

Both $\Phi$ functions are strictly increasing. With zero switching costs, $G(x^*)=0$ and strict monotonicity imply $\Phi_{\Atech}(x^*)<0<\Phi_{\HH}(x^*)$. Interior dominance boundaries therefore give $x_{\HH}^{E}<x^*<x_{\Atech}^{E}$. Continuity preserves a nonempty overlap for sufficiently small positive switching costs.
\end{proof}

\begin{proof}[Proof of Proposition~\ref{prop:firstbest}]
Differentiating equation~\eqref{eq:welfare}, with $b'(x)=\kappa$ and $L'=d\ell$, gives equation~\eqref{eq:welfareprime}; subtracting $G$ gives equation~\eqref{eq:externality}. The numerator of equation~\eqref{eq:marketx} is positive, increasing, and convex in $x$ because $W$ is convex and $L$ is affine, and the reciprocal of the denominator is positive, increasing, and convex because the denominator is positive and falls linearly. Their product $M$ is therefore convex, so $M'$ is nondecreasing, and $M'>0$ by equation~\eqref{eq:Mprime}. Since $b(x)>0$ is increasing and $G$ is strictly increasing, equation~\eqref{eq:welfareconvex} implies that $S'$ is strictly increasing. Hence $S$ is strictly convex and its maximizer is an endpoint.

The difference between endpoint welfare levels is
\begin{equation*}
    S(1)-S(0)=\kappa M_{\HH}
    +b_{\Atech}(M_{\HH}-M_{\Atech})
    -\bigl[C(\ell)-C(\ell(1-d))\bigr].
\end{equation*}
In the coordination case, $G_{\HH}>0$ implies $\kappa M_{\HH}>d\ell w_{\HH}$. Since $b_{\Atech}>0$, $M_{\HH}>M_{\Atech}$, and
\begin{equation*}
    C(\ell)-C(\ell(1-d))\leq d\ell w_{\HH},
\end{equation*}
the welfare difference is strictly positive. Strict convexity then makes $x=1$ the unique first best.
\end{proof}

\begin{proof}[Proof of Corollary~\ref{cor:resident}]
Since $W=w(L)L$ and $C'(L)=w(L)$, workers' participation rents satisfy $\frac{d}{dx}[W-C]=L\,w'(L)\,d\ell\geq0$, with strict inequality when $w'>0$. Proposition~\ref{prop:firstbest} makes $x=1$ the unique maximizer of $S$. Equation~\eqref{eq:residentwelfare} writes $R$ as a nonnegative combination of the two, so $x=1$ maximizes $R$, uniquely when $\theta>0$ or $w'>0$.
\end{proof}

\begin{proof}[Proof of Proposition~\ref{prop:dynplanner}]
The candidate values in equations~\eqref{eq:dynN}--\eqref{eq:dynA} satisfy
\begin{align}
    rH-S&=\rho(1-x)(H'-F_{\HH}), \label{eq:dynHidentity}\\
    rA-S&=-\rho x(A'+F_{\Atech}). \label{eq:dynAidentity}
\end{align}
Where $H\geq N$ and $x<1$, equation~\eqref{eq:dynHidentity} gives
$H'\geq F_{\HH}$. The automation term in equation~\eqref{eq:dynHJB}
therefore vanishes, and $H$ satisfies the full planner equation.
Where $A\geq N$ and $x>0$, equation~\eqref{eq:dynAidentity} gives
$A'\leq-F_{\Atech}$, and the corresponding result holds for $A$.

Integration by parts gives
\begin{align}
    H-N
    &=\frac{\rho(1-x)}r\int_0^\infty e^{-(r+\rho)t}
      [S'(x_t^{\HH})-rF_{\HH}]\,dt, \label{eq:dynHN}\\
    A-N
    &=-\frac{\rho x}r\int_0^\infty e^{-(r+\rho)t}
      [S'(x_t^{\Atech})+rF_{\Atech}]\,dt. \label{eq:dynAN}
\end{align}
Since $S'$ is strictly increasing and
$x_t^{\Atech}\leq x\leq x_t^{\HH}$, an interior state with
$N\geq H,A$ satisfies
\[
    -F_{\Atech}\leq \frac{S'(x)}r=N'(x)\leq F_{\HH}.
\]
Thus $N$ satisfies equation~\eqref{eq:dynHJB} wherever it is maximal.
At each endpoint the same reasoning applies to the available direction
of adjustment.

Both integrals in equations~\eqref{eq:dynHN}--\eqref{eq:dynAN} are strictly
increasing in $x$. Each adjustment value therefore crosses the inaction
value at most once in the interior. Moreover,
\begin{equation}
    (H-A)'=
    \int_0^\infty e^{-(r+\rho)t}
    [S'(x_t^{\HH})-S'(x_t^{\Atech})]\,dt
    +\frac{\rho(F_{\HH}+F_{\Atech})}{r+\rho}>0.
    \label{eq:dynHAmonotone}
\end{equation}
The two adjustment values also cross at most once. The function
$\widehat V=\max\{N,H,A\}$ is consequently bounded, Lipschitz, and
piecewise continuously differentiable, and it satisfies
equation~\eqref{eq:dynHJB} at every differentiability point.

For any admissible control and its absolutely continuous state path, the
planner equation implies
\[
    \frac{d}{dt}[e^{-rt}\widehat V(x_t)]
    \leq-e^{-rt}[S(x_t)-\mathcal K(x_t,u_t)]
\]
almost everywhere. If a path occupies a nondifferentiable state for
positive time, its derivative is zero almost everywhere on that set.
There the same inequality follows from $r\widehat V\geq S$ and
$\mathcal K\geq0$. Integrating to $T$ gives
\[
    \int_0^T e^{-rt}[S(x_t)-\mathcal K(x_t,u_t)]\,dt
    \leq \widehat V(x_0)-e^{-rT}\widehat V(x_T).
\]
Boundedness makes the last term vanish as $T\to\infty$, providing an
upper bound on every admissible payoff. Each candidate defining
$\widehat V$ is attained by its corresponding feasible policy, proving
equation~\eqref{eq:dynenvelope}.

The crossing properties place automation, inaction, and recovery in
that order. Once a strictly preferred adjustment path starts, it
remains preferred as the state moves in its direction. If inaction and
an adjustment path tie at the initial state, waiting there for any
finite duration and then starting that path preserves the value:
the surplus earned while waiting equals the discount cost of delaying
the continuation value. An automation--recovery tie strictly above
$N$ instead makes immediate adjustment optimal.
\end{proof}

\begin{proof}[Proof of Corollary~\ref{cor:dynregions}]
Substituting the quadratic surplus into
equations~\eqref{eq:dynH}--\eqref{eq:dynA} gives
\begin{align}
    A(x)&=\frac{S(0)}r+\frac{\Omega x^2}{2(r+2\rho)}
    -\frac{(\Omega x^S+\rho F_{\Atech})x}{r+\rho},
    \label{eq:dynAclosed}\\
    H(x)&=\frac{S(1)}r
    -\frac{[\Omega(1-x^S)+\rho F_{\HH}](1-x)}{r+\rho}
    +\frac{\Omega(1-x)^2}{2(r+2\rho)}.
    \label{eq:dynHclosed}
\end{align}
Their differences from inaction factor as
\[
    A-N=\frac{\Omega\rho x}{r(r+2\rho)}(a-x),
    \qquad
    H-N=\frac{\Omega\rho(1-x)}{r(r+2\rho)}(x-h),
\]
where $a$ and $h$ are defined in equation~\eqref{eq:dynboundaries}.
Their distance is
\[
    h-a=
    \frac{r(r+2\rho)(F_{\Atech}+F_{\HH})/\Omega-\rho}{r+\rho}.
\]
If $a\leq h$, automation dominates inaction below $a$, inaction
dominates both adjustment paths between $a$ and $h$, and recovery
dominates inaction above $h$. These regions are intersected with
the feasible state space.

If $h<a$, at least one adjustment value exceeds $N$ at every
interior state. Direct subtraction gives
\begin{equation}
    H-A=\frac{\rho}{r+\rho}
    \left\{
    \frac{\Omega[rx-(r+2\rho)x^S+\rho]}{r(r+2\rho)}
    +(F_{\HH}+F_{\Atech})x-F_{\HH}
    \right\}.
    \label{eq:dynvaluedifference}
\end{equation}
This expression is strictly increasing and has the unique algebraic
zero in equation~\eqref{eq:dynswitch}. Proposition~\ref{prop:dynplanner}
therefore gives the claimed choice on each side of that zero.
If the zero is interior, both adjustment paths there exceed $N$,
and their tie lies strictly between $h$ and $a$. Endpoint decisions
follow directly from equation~\eqref{eq:dynenvelope}.

At an interior automation--recovery threshold, raising $F_{\HH}$
reduces $H-A$ by $\rho(1-x)/(r+\rho)$ per unit of cost and therefore
raises the threshold. Raising $F_{\Atech}$ increases $H-A$ by
$\rho x/(r+\rho)$ and lowers the threshold. With zero switching costs,
equation~\eqref{eq:dynswitch} gives
\[
    c_0=\frac{(r+2\rho)x^S-\rho}{r},
    \qquad
    c_0-x^S=\frac{\rho}{r}(2x^S-1)<0,
\]
using $S(1)>S(0)$.

Finally, equation~\eqref{eq:dynrecoveryall} is equivalent to
$h\leq0$, hence to $H(0)\geq N(0)$. It gives $H\geq N$ everywhere.
Since $A(0)=N(0)$, equation~\eqref{eq:dynHAmonotone} then gives
$H>A$ at every positive state. Recovery is consequently optimal
from every initial state. Necessity follows by comparing recovery
with inaction at zero.
\end{proof}

\begin{proof}[Proof of Proposition~\ref{prop:dynimplementation}]
For a given anticipated aggregate path, let $J_{\HH}$ and $J_{\Atech}$
be the firm's values in the two modes and write $D=J_{\HH}-J_{\Atech}$.
At a revision opportunity an automated firm rebuilds when
$D>F_{\HH}$, and an augmenting firm automates when $D<-F_{\Atech}$.
Subtracting the two private value equations gives
\begin{equation}
    rD=G+s+\dot D
       +\rho[-D-F_{\Atech}]_+
       -\rho[D-F_{\HH}]_+.
    \label{eq:dynprivategap}
\end{equation}
The derivative is taken along the anticipated path, including the
passage of time while the aggregate state is constant.

On a smooth optimal branch let $P=V'$. Differentiating
equation~\eqref{eq:dynHJB} and applying the envelope theorem gives
\begin{equation}
    rP=S'+\dot P
       +\rho[-P-F_{\Atech}]_+
       -\rho[P-F_{\HH}]_+.
    \label{eq:dyncostate}
\end{equation}
Here $\dot P=fP'$ along the chosen branch. The two option-value
terms arise from the changing populations of automated and
augmenting firms eligible to switch.
Under $s=s^P=S'-G$, equations~\eqref{eq:dynprivategap} and
\eqref{eq:dyncostate} coincide.

The bounded solution along a given path is unique. Indeed, after
solving equation~\eqref{eq:dynprivategap} for $\dot D$, its right
side is strictly increasing in $D$ with slope at least $r$.
The distance between two distinct solutions therefore grows at
least exponentially at rate $r$. Boundedness rules out that
possibility. Hence $D=P$, and the private thresholds coincide
with the planner's switching thresholds.

At an initial tie between the two adjustment paths, use the
one-sided derivative of the selected branch. For perpetual
inaction at an interior state, $D=S'/r$ lies between
$-F_{\Atech}$ and $F_{\HH}$.
At an interior inaction--recovery tie with recovery scheduled
to start at time $T$, the value gap while firms wait is
\[
    D(t)=\frac{S'(x)}r+
    \left(F_{\HH}-\frac{S'(x)}r\right)e^{-r(T-t)},
    \qquad t<T.
\]
It remains in the inaction interval and reaches $F_{\HH}$ at $T$.
At an interior inaction--automation tie, the same construction
replaces $F_{\HH}$ by $-F_{\Atech}$.

At a stationary endpoint, take the unique constant solution of
equation~\eqref{eq:dynprivategap}. Optimal inaction at $x=0$
implies $S'(0)\leq rF_{\HH}$, so this solution has
$D\leq F_{\HH}$. At $x=1$, optimal inaction implies
$S'(1)\geq-rF_{\Atech}$, so $D\geq-F_{\Atech}$.
These are the switching constraints for the firms present at
the respective endpoints. For an optimal endpoint waiting spell
ending at $T$, solve equation~\eqref{eq:dynprivategap} backwards
from the relevant switching threshold at $T$. At $D=F_{\HH}$
the backward motion points toward $D\leq F_{\HH}$ when
$S'(0)\leq rF_{\HH}$. At $D=-F_{\Atech}$ it points toward
$D\geq-F_{\Atech}$ when $S'(1)\geq-rF_{\Atech}$.
Uniqueness of the scalar value equation keeps the solution
on that side of the available threshold throughout the wait.
The option terms account for a hypothetical firm's choices
in the other mode. Firms therefore optimally follow the
selected path at the endpoints as well as in the interior.
The balanced transfers preserve aggregate profit and local
spending, so the market path used in these value equations
is the one induced by their choices.
\end{proof}

\section{Expectations with switching costs and equilibrium selection}
\label{app:expectations}

\subsection{The expectations interval with switching costs}

This section uses the fixed-wage benchmark with outside owners, in which $G(x)=\Gamma(x-x^*)$ by equation~\eqref{eq:affineG}, and writes $\Xi=(r+2\rho)/(r+\rho)$.

\begin{corollary}[Closed-form expectations region]
\label{cor:closedexpectations}
Under equation~\eqref{eq:affineG}, define
\begin{align}
    x_{\Atech}^{E}
    &=\Xi x^*-\frac{(r+2\rho)\widehat F_{\Atech}}{\Gamma},
    \label{eq:xAE}\\
    x_{\HH}^{E}
    &=1-\Xi(1-x^*)+\frac{(r+2\rho)\widehat F_{\HH}}{\Gamma}.
    \label{eq:xHE}
\end{align}
For interior inherited shares, the automation path is an equilibrium exactly when $x\leq x_{\Atech}^{E}$, and the recovery path is an equilibrium exactly when $x\geq x_{\HH}^{E}$. The positive-length feasible overlap is
\begin{equation}
    \left[\min\{1,x_{\Atech}^{E}\}
    -\max\{0,x_{\HH}^{E}\}\right]_+,
    \qquad [y]_+\equiv\max\{y,0\}.
    \label{eq:feasibleexpectwidth}
\end{equation}
The constant paths at $x=0$ and $x=1$ obey the endpoint exceptions in Proposition~\ref{prop:expectations}; any resulting isolated overlap has zero interval length. When $0<x_{\HH}^{E}<x_{\Atech}^{E}<1$, equation~\eqref{eq:feasibleexpectwidth} equals the algebraic threshold distance
\begin{equation}
    x_{\Atech}^{E}-x_{\HH}^{E}
    =\frac{\rho}{r+\rho}
    -\frac{r(r+2\rho)(F_{\Atech}+F_{\HH})}
    {(r+\rho)\Gamma},
    \label{eq:expectwidth}
\end{equation}
which is positive if and only if
\begin{equation}
    F_{\Atech}+F_{\HH}<\frac{\Gamma\rho}{r(r+2\rho)}.
    \label{eq:Fcondition}
\end{equation}
Without interiority, condition~\eqref{eq:Fcondition} characterizes a positive algebraic distance; equation~\eqref{eq:feasibleexpectwidth} gives the corresponding feasible width.
The overlap straddles $x^*$ when
\begin{equation}
    F_{\Atech}<\frac{\Gamma\rho x^*}{r(r+2\rho)},
    \qquad
    F_{\HH}<\frac{\Gamma\rho(1-x^*)}{r(r+2\rho)}.
    \label{eq:straddle}
\end{equation}
\end{corollary}

\begin{proof}[Proof of Corollary~\ref{cor:closedexpectations}]
Substituting equation~\eqref{eq:affineG} into equations~\eqref{eq:PhiA}--\eqref{eq:PhiH} gives
\begin{align*}
    \Phi_{\Atech}(x)
    &=\Gamma\left(\frac{x}{r+2\rho}-\frac{x^*}{r+\rho}\right),\\
    \Phi_{\HH}(x)
    &=\Gamma\left(\frac{1-x^*}{r+\rho}
      -\frac{1-x}{r+2\rho}\right).
\end{align*}
Solving $\Phi_{\Atech}=-\widehat F_{\Atech}$ and $\Phi_{\HH}=\widehat F_{\HH}$ yields equations~\eqref{eq:xAE}--\eqref{eq:xHE}. Strict monotonicity and the endpoint qualifications in Proposition~\ref{prop:expectations} give the equilibrium sets stated in the corollary. Subtracting the two algebraic thresholds gives equation~\eqref{eq:expectwidth}; its positivity is equivalent to equation~\eqref{eq:Fcondition}. Intersecting the algebraic interval $[x_{\HH}^{E},x_{\Atech}^{E}]$ with $[0,1]$ gives equation~\eqref{eq:feasibleexpectwidth}. Comparing each threshold with $x^*$ gives equation~\eqref{eq:straddle}.
\end{proof}

\subsection{Equilibrium selection under aggregate shocks}
\label{sec:selection}

A public fundamental that moves continuously between dominance regions resolves the multiplicity in Proposition~\ref{prop:expectations}. In this subsection switching costs are zero. An equilibrium is a strategy profile under which, at each Poisson revision opportunity, a firm chooses the mode that maximizes its expected discounted payoff, taking as given the public diffusion and the aggregate law of motion induced by the other firms' strategies. Consider the nested affine game
\begin{equation}
    G(x,z)=\Gamma(x-x^*)+z,
    \qquad dz_t=\sigma\,dB_t,\quad \sigma>0.
    \label{eq:stochasticG}
\end{equation}
The shifter $z$ may represent profitability, local demand, or the relative capability of human oversight. In the underlying symmetric two-action game, the relative payoff against an automated opponent is $G(0,z)$ and the relative payoff against an augmenting opponent is $G(1,z)$. Affinity gives
\begin{equation}
    (1-x)G(0,z)+xG(1,z)=G(x,z),
    \label{eq:affinegamemapping}
\end{equation}
which maps the population share directly into the expected relative payoff. The game satisfies the assumptions of Theorem~1 in \citet{frankel2000resolving}, with $1+G$ in the role of their relative productivity: firms revise at the common Poisson rate $\rho$ and discount at rate $r>0$; $G$ is continuously differentiable and strictly increasing in both $x$ and $z$; the fundamental is a driftless Brownian motion; and augmentation is dominant whenever $z>\Gamma x^*$ and automation whenever $z<-\Gamma(1-x^*)$, levels that the diffusion reaches with probability one. The same properties satisfy the assumptions of their Theorem~3 and of Theorem~1 in \citet{burdzy2001fast}, which concern frequent revision.

\begin{proposition}[Selection under aggregate shocks]
\label{prop:stochasticselection}
In the environment just described, for every $\sigma>0$ the equilibrium is unique: there is a decreasing boundary $Z(x)$ such that a revising firm augments when $z>Z(x)$ and automates when $z<Z(x)$, so with finite revision frictions inherited history continues to matter. For almost every realized path of $z_t$, the equilibrium path of $x_t$ is unique. As $\rho\to\infty$, the risk-dominant action is selected: the boundary converges to
\begin{equation}
    z^{\mathrm{RD}}=\Gamma\left(x^*-\frac12\right).
    \label{eq:zRD}
\end{equation}
At $z=0$, augmentation is selected if $x^*<1/2$, and automation is selected if $x^*>1/2$.
\end{proposition}

In the affine game, the criterion follows from
\begin{equation}
    \int_0^1 G(x,z^{\mathrm{RD}})\,dx=0.
    \label{eq:potentialcriterion}
\end{equation}
It coincides with larger-basin selection under the myopic dynamics: augmentation's basin has length $1-x^*$ and automation's has length $x^*$.

\begin{proof}[Proof of Proposition~\ref{prop:stochasticselection}]
Under the assumptions verified in the text, Theorem~1 of \citet{frankel2000resolving}, applied with relative productivity $1+G$, gives the unique equilibrium and its decreasing boundary for every $\sigma>0$, and their Lemma~1 gives a unique path of $x_t$ for almost every path of $z_t$. As $\rho\to\infty$, their Theorem~3 gives a vertical boundary at the value of $z$ that solves equation~\eqref{eq:potentialcriterion}. By equation~\eqref{eq:affinegamemapping}, at that value a firm is indifferent when its opponent is equally likely to play either mode, which is the risk-dominance criterion of \citet[Theorem~1]{burdzy2001fast}. The indifference boundary solves
\begin{equation*}
    0=\int_0^1[\Gamma(x-x^*)+z^{\mathrm{RD}}]\,dx
    =\Gamma\left(\frac12-x^*\right)+z^{\mathrm{RD}},
\end{equation*}
which proves equation~\eqref{eq:zRD}. In the fixed-wage benchmark with outside owners, setting $x^*=1/2$ in $G(x^*)=0$ yields equation~\eqref{eq:dRD}. The inequalities $d_1<d_{\mathrm{RD}}<d_2$ follow because $1+\kappa\mu>1+\kappa\mu/2>1$.
\end{proof}

\section{Temporary support for recovery}
\label{app:recovery_support}

\subsection{Firms responding to current payoffs}
\label{app:myopicsupport}

The Pigouvian wedge equates private and social flow gains. Starting a recovery from a low augmentation share can require more. If $S'(0)<0$, the gain $G+s^P=S'$ is negative at low shares even though $S(1)>S(0)$, and firms that respond to current payoffs keep automating. A temporary addition to the wedge then starts the recovery. We characterize the smallest one under the myopic adjustment process in equation~\eqref{eq:myopicdynamics}, in which firms maximize current flow payoffs and switching costs are absent. Let $x^S$ be the unique solution to $S'(x^S)=0$. The smallest addition to $s^P$ that makes augmentation weakly optimal at share $x$ is
\begin{equation}
    \beta^P(x)=[-S'(x)]_+.
    \label{eq:welfarebridge}
\end{equation}
We call $\beta^P$ \emph{temporary support}. Since $S'$ is increasing, it declines as the augmentation share rises and ends when the share reaches $x^S$. A constant payment $\beta_0>-S'(x_0)$ from an inherited share $x_0<x^S$ until $x^S$ is crossed is also sufficient.

If the government seeks only to move the economy to the augmentation equilibrium, without maintaining the Pigouvian wedge, the smallest wedge is the \emph{selection wedge}
\begin{equation}
    s^{\mathrm{sel}}(x)=[-G(x)]_+.
    \label{eq:selectionwedge}
\end{equation}
It ends when the share crosses the static threshold $x^*$, and a constant wedge slightly above $-G(x_0)$ from $x_0<x^*$ until $x^*$ is crossed is also sufficient. Starting from the automated allocation, the selection wedge begins at
\begin{equation}
    s^{\mathrm{sel}}(0)=d\ell w_{\Atech}-\kappa M_{\Atech}.
    \label{eq:bridgezero}
\end{equation}
This is the amount a single firm loses by keeping its full team when every other firm has automated. Temporary support and the selection wedge are paid through the balanced transfer, and the Pigouvian wedge, where it is maintained, remains in force after temporary support ends.

\begin{proposition}[Pigouvian wedge and minimal temporary support]
\label{prop:policy}
Suppose the economy is in the coordination case and consider the myopic adjustment process in equation~\eqref{eq:myopicdynamics}, with firms maximizing current flow payoffs and with no switching costs. Equation~\eqref{eq:pigou} is the unique state-contingent wedge that makes the private flow gain equal the planner's marginal return. Conditional on maintaining this wedge, equation~\eqref{eq:welfarebridge} is the pointwise-minimal nonnegative addition that makes augmentation weakly optimal at every share and, with ties broken in favor of augmentation, induces a weakly monotone transition to the local first best. If a constant addition is used from $x_0<x^S$ until $x^S$ is crossed, the infimum of the sufficient constants is $-S'(x_0)$. If equilibrium selection alone is the objective and the Pigouvian wedge is not maintained, equation~\eqref{eq:selectionwedge} is the corresponding pointwise-minimal nonnegative wedge. Arbitrarily small strict increments deliver strict convergence in either case.
\end{proposition}

\begin{proof}[Proof of Proposition~\ref{prop:policy}]
Under myopic adjustment and in the absence of switching costs, a firm chooses $\HH$ whenever its current relative flow payoff is weakly positive. Equation~\eqref{eq:externality} gives $G+s^P=S'$. Conditional on the Pigouvian wedge, a nonnegative addition $\beta$ therefore makes $\HH$ weakly optimal at share $x$ exactly when $\beta(x)\geq-S'(x)$. With ties broken in favor of augmentation, the pointwise-minimal such addition is $[-S'(x)]_+$. Adding any $\varepsilon>0$ when $S'(x)\leq0$ makes the preference strict. Equation~\eqref{eq:welfareconvex} makes $S'$ strictly increasing. Hence, if $x_0<x^S$, a constant addition $\beta_0>-S'(x_0)$ remains sufficient until $x^S$ is crossed, and the infimum over such constants is $-S'(x_0)$. If the Pigouvian wedge is not maintained, the corresponding flow condition is $s(x)\geq-G(x)$; intersecting it with nonnegativity gives equation~\eqref{eq:selectionwedge}.
\end{proof}

These results describe how to support a recovery; whether a recovery repays its rebuilding costs from a given inherited share is the question answered by Proposition~\ref{prop:dynplanner}. Policy differs outside the coordination case. When automation is each firm's dominant choice, $G_{\HH}<0$, the automated allocation is the unique equilibrium, as in \citet{falk2026layoff}, and the augmented allocation is still the local first best whenever $S(1)>S(0)$. Augmentation then survives only under a wedge that is maintained, since without one it is not an equilibrium. In the coordination case the augmentation equilibrium sustains itself: temporary support ends, and the balanced transfer collects nothing once every firm has a team. If instead $S(0)\geq S(1)$, the automated allocation is itself the local first best.

\subsection{Withdrawal with forward-looking firms}
\label{app:withdrawal}

This section uses outside ownership when switching costs are positive. Its zero-cost results apply at every ownership share.

Forward-looking firms evaluate the entire anticipated payoff path. Let $s_0(x)$ denote a continuous wedge that remains after temporary support is withdrawn, implemented through the balanced transfer of Section~\ref{sec:welfare} so that it does not change $M(x)$. Define the post-withdrawal flow gain
\begin{equation}
    \widetilde G(x)=G(x)+s_0(x).
    \label{eq:postwithdrawalgain}
\end{equation}
Suppose $\widetilde G$ is strictly increasing. Along the automation and recovery paths, the corresponding present-value gains are
\begin{align}
    \widetilde\Phi_{\Atech}^{\,s_0}(x)
    &=\int_0^\infty e^{-(r+\rho)t}
      \widetilde G\bigl(xe^{-\rho t}\bigr)\,dt,
    \label{eq:policyPhiA}\\
    \widetilde\Phi_{\HH}^{\,s_0}(x)
    &=\int_0^\infty e^{-(r+\rho)t}
      \widetilde G\bigl(1-(1-x)e^{-\rho t}\bigr)\,dt.
    \label{eq:policyPhiH}
\end{align}
For $x\in(0,1]$, the automation path is an equilibrium after withdrawal exactly when $\widetilde\Phi_{\Atech}^{\,s_0}(x)\leq-\widehat F_{\Atech}$. For $x\in[0,1)$, the recovery path is an equilibrium after withdrawal exactly when $\widetilde\Phi_{\HH}^{\,s_0}(x)\geq\widehat F_{\HH}$. At the stationary endpoints, the automated allocation is an equilibrium if and only if $\widetilde G(0)\leq rF_{\HH}$, and the augmented allocation is an equilibrium if and only if $\widetilde G(1)\geq-rF_{\Atech}$. Temporary support can therefore be withdrawn, keeping the recovery path and eliminating the automation path, once both present-value conditions hold in the required directions.

With the Pigouvian wedge maintained, $s_0=s^P$ and
\begin{equation}
    \widetilde G(x)=G(x)+s^P(x)=S'(x).
    \label{eq:correctedgain}
\end{equation}

\begin{proposition}[Withdrawal of temporary support]
\label{prop:forwardpolicy}
\leavevmode\\
Suppose the conditions of Proposition~\ref{prop:firstbest} hold, with outside ownership when switching costs are positive, and maintain the Pigouvian wedge $s^P$ after temporary support is withdrawn.
\begin{enumerate}
    \item For every $F_{\Atech}\geq0$, there exists $\bar x_{\Atech}<1$ such that the automation path is not an equilibrium from any inherited share $x>\bar x_{\Atech}$.
    \item There exists $\bar x_{\HH}<1$ such that the recovery path is a strict equilibrium from every inherited share $x\in(\bar x_{\HH},1)$ if and only if
    \begin{equation}
        rF_{\HH}<S'(1).
        \label{eq:rebuildingfeasibility}
    \end{equation}
\end{enumerate}
Under condition~\eqref{eq:rebuildingfeasibility}, temporary support can be withdrawn once the recovery path reaches any share
\[
    \bar x>\max\{\bar x_{\Atech},\bar x_{\HH}\},
    \qquad \bar x<1.
\]
The condition holds automatically when $F_{\HH}=0$ and for all sufficiently small rebuilding costs.
\end{proposition}

The fixed-wage benchmark with outside owners makes these corrected boundaries explicit. Write $S'(x)=\Omega(x-x^S)$, where $\Omega=2\kappa\mu dK>0$ and $x^S$ denotes the algebraic zero of $S'$ whether or not it lies in $[0,1]$, and retain $\Xi=(r+2\rho)/(r+\rho)$. When the algebraic boundaries are interior,
\begin{align}
    \widetilde x_{\Atech}^{E}
    &=\Xi x^S
      -\frac{(r+2\rho)\widehat F_{\Atech}}{\Omega},
    \label{eq:correctedxAE}\\
    \widetilde x_{\HH}^{E}
    &=1-\Xi(1-x^S)
      +\frac{(r+2\rho)\widehat F_{\HH}}{\Omega}.
    \label{eq:correctedxHE}
\end{align}
From $S(1)-S(0)=\Omega(1/2-x^S)>0$, the first-best ranking implies $x^S<1/2$. Since $\Xi<2$, equation~\eqref{eq:correctedxAE} gives $\widetilde x_{\Atech}^{E}<1$ for every $F_{\Atech}\geq0$. Equation~\eqref{eq:correctedxHE} lies below one exactly under condition~\eqref{eq:rebuildingfeasibility}. Boundaries outside $[0,1]$ are interpreted after clipping to the feasible state space.

If every wedge is instead withdrawn, then $s_0=0$ and the relevant boundaries revert to $x_{\HH}^{E}$ and $x_{\Atech}^{E}$ from Corollary~\ref{cor:closedexpectations}. The recovery path remains an equilibrium once $x>x_{\HH}^{E}$, and the automation path is eliminated once $x>x_{\Atech}^{E}$. With zero switching costs, a share below one that satisfies both conditions exists if and only if
\begin{equation}
    x^*<\frac{r+\rho}{r+2\rho}.
    \label{eq:temporaryselectionfeasibility}
\end{equation}
In that case temporary support that raises the augmentation share above $x_{\Atech}^{E}$ eliminates the automation path with no wedge left in force; otherwise a wedge must be maintained.

For $F_{\HH}=0$, the temporary support in equation~\eqref{eq:welfarebridge}, together with the maintained Pigouvian wedge, supports the recovery path under the stated tie-break, because the flow gain while it is paid is $\max\{S'(x),0\}$; anticipated future gains allow smaller support. With $F_{\HH}>0$, temporary support must also satisfy the forward-looking rebuilding condition. Proposition~\ref{prop:dynplanner} ranks the feasible reorganization plans net of their real costs, and Proposition~\ref{prop:dynimplementation} supports the optimal path.

In the vanishing-friction limit with aggregate shocks, the switching boundary no longer depends on $x$, so policy that favors augmentation acts on the payoff fundamental, the demand feedback, or the relative return $\kappa$.

\begin{proof}[Proof of Proposition~\ref{prop:forwardpolicy}]
Replacing $G$ by the post-withdrawal gain $\widetilde G$ in the proof of Proposition~\ref{prop:expectations} gives the two interior path conditions following equations~\eqref{eq:policyPhiA}--\eqref{eq:policyPhiH}. Under the maintained Pigouvian wedge, $\widetilde G=S'$ is continuous and strictly increasing.

At $x=1$,
\begin{equation*}
    \widetilde\Phi_{\Atech}^{\,s^P}(1)
    =\int_0^\infty e^{-(r+\rho)t}
      S'\bigl(e^{-\rho t}\bigr)\,dt
    =\frac{1}{\rho}\int_0^1 u^{r/\rho}S'(u)\,du.
\end{equation*}
The functions $u^{r/\rho}$ and $S'(u)$ are increasing. Chebyshev's integral inequality and Proposition~\ref{prop:firstbest} therefore give
\begin{align*}
    \widetilde\Phi_{\Atech}^{\,s^P}(1)
    &\geq \frac{1}{\rho}
      \left(\int_0^1u^{r/\rho}\,du\right)
      \left(\int_0^1S'(u)\,du\right)\\
    &=\frac{S(1)-S(0)}{r+\rho}>0.
\end{align*}
Since $-\widehat F_{\Atech}\leq0$, the automation-path condition fails at $x=1$. Continuity and strict monotonicity of $\widetilde\Phi_{\Atech}^{\,s^P}$ give a cutoff strictly below one above which it continues to fail.

For the limiting recovery path as $x\uparrow1$,
\begin{equation*}
    \widetilde\Phi_{\HH}^{\,s^P}(1)
    =\frac{S'(1)}{r+\rho}.
\end{equation*}
This exceeds $\widehat F_{\HH}=rF_{\HH}/(r+\rho)$ exactly when $rF_{\HH}<S'(1)$. Continuity then gives an interior cutoff above which the recovery-path condition is strict. Every cutoff below one is reached in finite time along $x_t=1-(1-x_0)e^{-\rho t}$, which proves the withdrawal claim.
\end{proof}

\section{Resident ownership and optimal reorganization}
\label{app:ownplanner}

Resident owners finance their share of organizational expenditure from income
that would otherwise support local spending. Let
\[
 a_0(x)=1-\mu\theta b(x)>0,\qquad
 M_0(x)=\frac{\overline M+\mu(1-\theta)W(x)}{a_0(x)},\qquad
 S_0(x)=b(x)M_0(x)-C(L(x)).
\]
The subscript $0$ denotes market size and surplus before current reorganization
expenditure. With controls $u=(u_{\HH},u_{\Atech})$, define
\[
 f(x,u)=\rho[(1-x)u_{\HH}-xu_{\Atech}],\qquad
 \mathcal K(x,u)=\rho[(1-x)u_{\HH}F_{\HH}+xu_{\Atech}F_{\Atech}].
\]
The resource costs are payments for organizational inputs supplied outside the
region. Net firm profits are $\Pi_u=b(x)M_u-W(x)-\mathcal K(x,u)$.
Substitution in the local spending equation gives
\begin{align}
 M_u(x)&=M_0(x)-\frac{\mu\theta\mathcal K(x,u)}{a_0(x)},
 \label{eq:ownmarket}\\
 S_u(x)&=b(x)M_u(x)-C(L(x))-\mathcal K(x,u)
       =S_0(x)-\frac{\mathcal K(x,u)}{a_0(x)}.
 \label{eq:ownsurplus}
\end{align}
The external balance is
\begin{equation}
 \overline M=(1-\mu)[W+\theta\Pi_u]+(1-b)M_u
                 +\mathcal K+(1-\theta)\Pi_u.
 \label{eq:ownaccounts}
\end{equation}
Thus organizational imports and the reduction in dividends enter the same
accounts as model usage and service exports. Assume that firms can finance
reorganization and that consumption budgets remain positive and incumbents
continue operating for every $x\in[0,1]$ and $u\in[0,1]^2$ considered below.

The social cost of a switch is
$c_j(x)=F_j/a_0(x)$, including the associated contraction in spending.
The planner's value $V_\theta$ satisfies
\begin{equation}
 rV_\theta=S_0+\rho(1-x)[V_\theta'-c_{\HH}(x)]_+
                 +\rho x[-V_\theta'-c_{\Atech}(x)]_+.
 \label{eq:ownHJB}
\end{equation}
For the paths $x_t^{\HH}=1-(1-x)e^{-\rho t}$ and
$x_t^{\Atech}=xe^{-\rho t}$, let
\begin{align}
 N_\theta(x)&=\frac{S_0(x)}r, \label{eq:ownN}\\
 H_\theta(x)&=\int_0^\infty e^{-rt}
 \left[S_0(x_t^{\HH})
       -\rho(1-x_t^{\HH})c_{\HH}(x_t^{\HH})\right]\,dt,
 \label{eq:ownH}\\
 A_\theta(x)&=\int_0^\infty e^{-rt}
 \left[S_0(x_t^{\Atech})
       -\rho x_t^{\Atech}c_{\Atech}(x_t^{\Atech})\right]\,dt.
 \label{eq:ownA}
\end{align}

\begin{samepage}
\begin{proposition}[Optimal reorganization with resident ownership]
\label{prop:ownplanner}
Suppose $S_0$ is twice continuously differentiable and strictly convex.
If both $S_0'-rc_{\HH}$ and $S_0'+rc_{\Atech}$ are strictly increasing, then
\begin{equation}
 V_\theta(x)=\max\{N_\theta(x),H_\theta(x),A_\theta(x)\}.
 \label{eq:ownenvelope}
\end{equation}
The optimal regions are ordered as automation, inaction, and recovery,
with any of the regions possibly empty.
A sufficient condition for the required monotonicity is
\begin{equation}
 S_0''(x)>
 \frac{rF_{\HH}\mu\theta\kappa}{a_0(x)^2}
 \qquad\text{for every }x\in[0,1].
 \label{eq:owncurvature}
\end{equation}
\end{proposition}
\end{samepage}

\begin{proof}
Define $Q_{\HH}=S_0'-rc_{\HH}$ and $Q_{\Atech}=S_0'+rc_{\Atech}$.
The candidate adjustment values obey
\begin{align*}
 rH_\theta-S_0&=\rho(1-x)(H_\theta'-c_{\HH}),\\
 rA_\theta-S_0&=-\rho x(A_\theta'+c_{\Atech}).
\end{align*}
Where $H_\theta\geq N_\theta$ and $x<1$, the first identity gives
$H_\theta'\geq c_{\HH}$, so the recovery candidate satisfies
the full equation~\eqref{eq:ownHJB}. Where
$A_\theta\geq N_\theta$ and $x>0$, the second gives
$A_\theta'\leq-c_{\Atech}$ and verifies the automation candidate.
Integration by parts gives
\begin{align*}
 H_\theta-N_\theta
 &=\frac{\rho(1-x)}r\int_0^\infty
        e^{-(r+\rho)t}Q_{\HH}(x_t^{\HH})\,dt,\\
 A_\theta-N_\theta
 &=-\frac{\rho x}r\int_0^\infty
        e^{-(r+\rho)t}Q_{\Atech}(x_t^{\Atech})\,dt.
\end{align*}
If $N_\theta$ is largest at an interior state, monotonicity of the two
$Q$ functions and $x_t^{\Atech}\leq x\leq x_t^{\HH}$ imply
$-c_{\Atech}\leq N_\theta'\leq c_{\HH}$.
The inaction candidate therefore satisfies the same maximized equation.
At an endpoint the corresponding argument applies to the available direction.

The two integrals are strictly increasing in the initial state.
Consequently recovery can cross inaction only once from below, and automation
can cross it only once from above. On an interior interval where both
adjustment values strictly exceed $N_\theta$,
\[
 (H_\theta-A_\theta)'=c_{\HH}+c_{\Atech}
 +\frac{rH_\theta-S_0}{\rho(1-x)}
 +\frac{rA_\theta-S_0}{\rho x}>0.
\]
There is at most one crossing relevant to the maximum of the two adjustment
values. Their maximum with $N_\theta$ is bounded, Lipschitz, and piecewise
continuously differentiable. The verification argument in
Proposition~\ref{prop:dynplanner} now bounds every admissible payoff by
this maximum, and each candidate attains its own value.
The crossing properties establish the ordering of the optimal regions.
Finally, $c_j'=F_j\mu\theta\kappa/a_0^2$.
Condition~\eqref{eq:owncurvature} makes $Q_{\HH}$ strictly increasing,
and strict convexity of $S_0$ together with $c_{\Atech}'\geq0$
makes $Q_{\Atech}$ strictly increasing.
\end{proof}

\paragraph{Dynamic implementation.}
Let $P=V_\theta'$ on a selected optimal branch, and let $f$ be the
drift on that branch. Define $G_0=\kappa M_0-d\ell w(L)$.
The wedge, implemented by the balanced transfer,
\begin{equation}
 s_\theta=a_0S_0'-G_0-a_0'fP
 \label{eq:ownwedge}
\end{equation}
supports the selected planner path under the firm participation and financing
conditions. Its economic components are
\begin{equation}
 s_\theta=\mu b\,d\ell\,[w+(1-\theta)Lw']
                      +\mu\theta\kappa fP.
 \label{eq:ownwedgecomponents}
\end{equation}
The first term values the income retained in the local market.
The second adjusts incentives for the change in the financing multiplier
as the economy reorganizes. Both terms are nonnegative along an optimal
path, because recovery has $f\geq0$ and $P\geq c_{\HH}$, while automation
has $f\leq0$ and $P\leq-c_{\Atech}$.

To verify implementation, differentiating the maximized planner equation gives
\[
 rP=S_0'+fP'
 +\rho[-P-c_{\Atech}]_+-\rho[P-c_{\HH}]_+
 +\frac{\mathcal K a_0'}{a_0^2}.
\]
Set the private value difference to $D=a_0P$. Multiplication by $a_0$
and use of $D'=a_0P'+a_0'P$ yield
\[
 rD=a_0S_0'-a_0'fP+\frac{\mathcal K a_0'}{a_0}
 +fD'+\rho[-D-F_{\Atech}]_+-\rho[D-F_{\HH}]_+.
\]
By equation~\eqref{eq:ownmarket}, the actual private flow gain is
$G_u=G_0+a_0'\mathcal K/a_0$.
Substitution of equation~\eqref{eq:ownwedge} therefore gives the private
value-difference equation under the selected path. The boundedness argument
in Proposition~\ref{prop:dynimplementation} identifies $D=a_0P$.
The private switching conditions $D>F_{\HH}$ and $D<-F_{\Atech}$ then
coincide with the planner's conditions.
At a tie between adjustment paths, choose the one-sided marginal value
associated with the selected path. At interior inaction states, the private
value difference is $a_0S_0'/r$. During a waiting spell at an interior
inaction--adjustment tie, the waiting-spell construction in that proposition
applies with this stationary value. At a stationary endpoint, use the unique
constant solution of the private value-difference equation with corrected
flow gain $a_0S_0'$. For delayed departure from an endpoint, solve that
equation backward from the switching threshold at departure. At each endpoint
only the switching inequality of the mode present must hold during the wait.
Balanced payments preserve aggregate profit at each state and control, so
they also preserve equations~\eqref{eq:ownmarket}--\eqref{eq:ownaccounts}.
At $\theta=0$, equation~\eqref{eq:ownwedge} reduces to the Pigouvian wedge
$s^P$ of Proposition~\ref{prop:dynimplementation}.

\section{Per-job labor}
\label{app:perjob}

The modern technique in \citet{murphy1989industrialization} combines a fixed labor requirement with labor in proportion to output. This section adds the second component to both deployment modes. In addition to its team, an incumbent in mode $j$ employs $\alpha_j\geq0$ workers for each job it completes. The wage is fixed at $w_0$, as in the fixed-wage benchmark of Section~\ref{sec:static}, and $\bar N$ exceeds total employment in every allocation. Define
\begin{equation}
    b_j=1-\frac{v_j+w_0\alpha_j}{\lambda_j},
    \qquad
    \psi_j=\frac{w_0\alpha_j}{\lambda_j},
    \qquad
    \psi(x)=\psi_{\Atech}+(\psi_{\HH}-\psi_{\Atech})x.
    \label{eq:perjobdefs}
\end{equation}
The margin $b_j$ is the share of each sales dollar that remains after the incumbent pays for model usage and per-job wages, and $\psi_j$ is the share paid as per-job wages; with $\alpha_j=0$, $b_j$ is the margin in Lemma~\ref{lem:limitpricing}. Assumption~\ref{ass:quality} becomes $v_j+w_0\alpha_j<\lambda_j$ for $j\in\{\Atech,\HH\}$ and $(v_{\HH}+w_0\alpha_{\HH})/\lambda_{\HH}<(v_{\Atech}+w_0\alpha_{\Atech})/\lambda_{\Atech}$, so that $\kappa=b_{\HH}-b_{\Atech}>0$. We take $\psi_{\HH}\geq\psi_{\Atech}$: per sales dollar, augmentation pays at least as much in per-job wages as automation, in which the model takes over tasks. Assumption~\ref{ass:viability} is maintained with $b_{\Atech}$ as defined here.

The proof of Lemma~\ref{lem:limitpricing} applies with the cost per job $v_j+w_0\alpha_j$ in place of $v_j$, so the incumbent charges $p_j=\lambda_j$ and earns
\[
    \profit_{\HH}(M)=b_{\HH}M-w_0\ell,
    \qquad
    \profit_{\Atech}(M)=b_{\Atech}M-w_0(1-d)\ell.
\]
Team employment is $L(x)$ in equation~\eqref{eq:employment}, so the wage bill is $w_0L(x)+\psi(x)M$ and aggregate operating profit is $b(x)M-w_0L(x)$. Equation~\eqref{eq:demand} then gives
\begin{equation}
    M(x)=\frac{\overline M+\mu(1-\theta)w_0L(x)}{1-\mu[\psi(x)+\theta b(x)]}.
    \label{eq:perjobmarket}
\end{equation}
Per-job wages are spent locally, which adds $\mu\psi(x)$ to the local share of each sales dollar in the denominator. The denominator is positive because $\psi_j+\theta b_j\leq\psi_j+b_j=1-v_j/\lambda_j<1$. The gain from augmentation is again $G(x)=\kappa M(x)-d\ell w_0$. As in Section~\ref{sec:static}, let $K=w_0\ell$ and $q=M(1)/K$.

\begin{proposition}[Per-job labor]
\label{prop:perjob}
Under the conditions above:
\begin{enumerate}
    \item $M$ is strictly increasing and convex, and $G$ is strictly increasing.
    \item Corollary~\ref{cor:resident} and Propositions~\ref{prop:flexclassification}, \ref{prop:firmranking}, \ref{prop:expectations}, \ref{prop:firstbest}, \ref{prop:dynplanner}, \ref{prop:dynimplementation}, \ref{prop:policy}, and~\ref{prop:forwardpolicy} hold, with outside ownership wherever those results assume it.
    \item The cutoffs for the displacement share are
    \begin{equation}
        d_1=\frac{\kappa q\,(1-\mu\psi_{\HH}-\mu\theta b_{\HH})}
        {1+\kappa\mu-\mu\psi_{\Atech}-\mu\theta b_{\HH}},
        \qquad
        d_2=\kappa q,
        \label{eq:perjobcutoffs}
    \end{equation}
    with augmentation unique for $d<d_1$, coordination for $d_1<d<d_2$, and automation unique for $d>d_2$, within the parameter domain defined by Assumption~\ref{ass:viability} and $d\in(0,1]$. The distance between the cutoffs is
    \begin{equation}
        d_2-d_1=\frac{\kappa\mu q\,(\kappa+\psi_{\HH}-\psi_{\Atech})}
        {1+\kappa\mu-\mu\psi_{\Atech}-\mu\theta b_{\HH}}.
        \label{eq:perjobwidth}
    \end{equation}
    At a fixed ratio $q$, $d_1$ strictly decreases and the distance strictly increases in $\theta$, and Assumption~\ref{ass:viability} at $d=d_1$ ensures it for every $d\in[d_1,d_2]$. For given $\kappa$, $b_{\HH}$, $q$, $\mu$, and $\theta$, the distance strictly exceeds that in equation~\eqref{eq:ownershipwidth} whenever $\alpha_{\HH}>0$.
    \item With outside owners and $\psi_{\HH}=\psi_{\Atech}=\psi$, equation~\eqref{eq:perjobmarket} becomes $M(x)=[\overline M+\mu w_0L(x)]/(1-\mu\psi)$. Every result and closed form derived for the fixed-wage benchmark with outside owners then holds with $\overline M$ and $\mu$ replaced by $\overline M/(1-\mu\psi)$ and $\mu/(1-\mu\psi)$.
\end{enumerate}
\end{proposition}

Part~3 is the counterpart of Proposition~\ref{prop:ownership}, with the same upper cutoff. Per-job wages raise the share of each sales dollar that is spent locally, which amplifies the demand created by each team's payroll, and when $\psi_{\HH}>\psi_{\Atech}$ an augmenting firm also pays more of each sales dollar as wages. Both effects strengthen the demand channel and widen the coordination interval.

\begin{proof}[Proof of Proposition~\ref{prop:perjob}]
Let $D(x)=1-\mu[\psi(x)+\theta b(x)]$ denote the denominator of equation~\eqref{eq:perjobmarket}. The numerator is affine in $x$ with slope $\mu(1-\theta)d\ell w_0\geq0$, and $D$ is affine with slope $D'=-\mu(\psi_{\HH}-\psi_{\Atech}+\theta\kappa)\leq0$. Hence
\begin{equation*}
    M'(x)=\frac{\mu(1-\theta)d\ell w_0+\mu(\psi_{\HH}-\psi_{\Atech}+\theta\kappa)M(x)}{D(x)}>0,
    \qquad
    M''(x)=-\frac{2D'M'(x)}{D(x)}\geq0,
\end{equation*}
where $M'>0$ because the first term in the numerator is positive when $\theta<1$ and the second is positive when $\theta>0$. Then $G'=\kappa M'>0$, which proves part~1.

For part~2, the proofs of Propositions~\ref{prop:flexclassification} and~\ref{prop:expectations} use only the continuity and strict monotonicity of $G$. With a fixed wage, the residual-wage term in the proof of Proposition~\ref{prop:firmranking} vanishes and $b_{\Atech}(M_{\HH}-M_{\Atech})>0$. Every employed worker's wage equals the worker's participation cost, so workers' rents are zero and joint surplus is aggregate operating profit,
\begin{equation*}
    S(x)=b(x)M(x)-w_0L(x),
    \qquad
    S'(x)=G(x)+b(x)M'(x).
\end{equation*}
The adoption externality is again $b(x)M'(x)>0$. By part~1, $S'$ is strictly increasing, so $S$ is strictly convex and twice continuously differentiable, and in the coordination case
\begin{equation*}
    S(1)-S(0)=\kappa M_{\HH}+b_{\Atech}(M_{\HH}-M_{\Atech})-d\ell w_0>0,
\end{equation*}
because $G_{\HH}=\kappa M_{\HH}-d\ell w_0>0$. These are the properties of $S$ and $G$ that the proofs of Propositions~\ref{prop:firstbest}, \ref{prop:dynplanner}, \ref{prop:dynimplementation}, \ref{prop:policy}, and~\ref{prop:forwardpolicy} use. Residents' welfare is $\theta S$, which gives Corollary~\ref{cor:resident}.

For part~3, write $\xi_j=\mu(\psi_j+\theta b_j)$. Equation~\eqref{eq:perjobmarket} at $x=1$ gives $\overline M=qK(1-\xi_{\HH})-\mu(1-\theta)K$, and at $x=0$
\begin{equation*}
    M_{\Atech}=K\,\frac{q(1-\xi_{\HH})-\mu(1-\theta)d}{1-\xi_{\Atech}}.
\end{equation*}
The boundary $G_{\HH}=(\kappa q-d)K=0$ gives $d_2$. The boundary $G_{\Atech}=0$ gives $\kappa[q(1-\xi_{\HH})-\mu(1-\theta)d]=d(1-\xi_{\Atech})$. Using $1-\xi_{\Atech}+\kappa\mu(1-\theta)=1+\kappa\mu-\mu\psi_{\Atech}-\mu\theta b_{\HH}$, its solution is $d_1$ in equation~\eqref{eq:perjobcutoffs}, and $\xi_{\HH}-\xi_{\Atech}+\kappa\mu(1-\theta)=\mu(\kappa+\psi_{\HH}-\psi_{\Atech})$ gives equation~\eqref{eq:perjobwidth}. Since $q$ does not depend on $d$, both $G_{\Atech}$ and $G_{\HH}$ strictly decrease in $d$, and Proposition~\ref{prop:flexclassification} classifies the three ranges. At fixed $q$, write $\delta=\mu\theta b_{\HH}$; then
\begin{equation*}
    \frac{\partial}{\partial\delta}\,
    \frac{1-\mu\psi_{\HH}-\delta}{1+\kappa\mu-\mu\psi_{\Atech}-\delta}
    =-\frac{\mu(\kappa+\psi_{\HH}-\psi_{\Atech})}
    {(1+\kappa\mu-\mu\psi_{\Atech}-\delta)^2}<0,
\end{equation*}
so $d_1$ strictly decreases in $\theta$, and the distance, whose denominator falls as $\theta$ rises, strictly increases. Automated profit is lowest at $x=0$, where
\begin{equation*}
    \frac{\profit_{\Atech}(M_{\Atech})}{K}
    =b_{\Atech}\,\frac{q(1-\xi_{\HH})-\mu(1-\theta)d}{1-\xi_{\Atech}}-(1-d).
\end{equation*}
Its derivative with respect to $d$ is $1-\mu(1-\theta)b_{\Atech}/(1-\xi_{\Atech})>0$, because $\mu(1-\theta)b_{\Atech}+\xi_{\Atech}=\mu(1-v_{\Atech}/\lambda_{\Atech})<1$. Viability at $d=d_1$ therefore extends to every $d\in[d_1,d_2]$. Finally, the ratio of the distance in equation~\eqref{eq:perjobwidth} to that in equation~\eqref{eq:ownershipwidth} is
\begin{equation*}
    \frac{\kappa+\psi_{\HH}-\psi_{\Atech}}{\kappa}\cdot
    \frac{1+\kappa\mu-\mu\theta b_{\HH}}{1+\kappa\mu-\mu\psi_{\Atech}-\mu\theta b_{\HH}}.
\end{equation*}
Both factors are at least one. When $\alpha_{\HH}>0$, $\psi_{\HH}>0$, so either $\psi_{\Atech}<\psi_{\HH}$ and the first factor exceeds one, or $\psi_{\Atech}=\psi_{\HH}>0$ and the second does.

For part~4, $\theta=0$ and $\psi_{\HH}=\psi_{\Atech}$ make $D$ constant, which gives the stated form of $M$: equation~\eqref{eq:marketx} with outside owners and the substituted parameters. Profits, $G$, and $S$ depend on the primitives only through $b_{\Atech}$, $b_{\HH}$, $w_0$, $\ell$, $d$, and $M$, so every result and closed form derived from them in that benchmark holds after the substitution.
\end{proof}

\end{document}